\newif\ifconf
\ifconf
    \documentclass[UKenglish,cleveref,autoref,thm-restate]{lipics-v2021} 
    \author{Matthias Bentert}{University of Bergen, Norway}{matthias.bentert@uib.no}{}{Supported by the European Research Council (ERC) under the European Union’s Horizon 2020 research and innovation programme (grant agreement No. 819416).}
    \author{Fedor V. Fomin}{University of Bergen, Norway}{fedor.fomin@uib.no}{}{Supported by the Research Council of Norway via the project BWCA (grant no. 314528).}
    \author{Fanny Hauser}{Technische Universität Berlin, Germany \and University of Bergen, Norway}{f.hauser@tu-berlin.de}{}{Supported by the German Academic Exchange Service under the Erasmus+ program (Grant Agreement 2023\#009).}
    \author{Saket Saurabh}{The Institute of Mathematical Sciences, India \and University of Bergen, Norway}{saket@imsc.res.in}{}{Supported by the European Research Council (ERC) under the European Union’s Horizon 2020 research and innovation programme (grant agreement No. 819416).}
    \authorrunning{M. Bentert, F. V. Fomin, F. Hauser, and S. Saurabh} 
    \Copyright{Matthias Bentert, Fedor V. Fomin, Fanny Hauser, and Saket Saurabh}
    
    
    \EventEditors{John Q. Open and Joan R. Access}
    \EventNoEds{2}
    \EventLongTitle{42nd Conference on Very Important Topics (CVIT 2016)}
    \EventShortTitle{CVIT 2016}
    \EventAcronym{CVIT}
    \EventYear{2016}
    \EventDate{December 24--27, 2016}
    \EventLocation{Little Whinging, United Kingdom}
    \EventLogo{}
    \SeriesVolume{42}
    \ArticleNo{23}
    
\else
    \documentclass{article}

    \usepackage[english]{babel}
    \usepackage[letterpaper,top=2cm,bottom=2cm,left=3cm,right=3cm,marginparwidth=1.75cm]{geometry}

    \usepackage[numbers]{natbib}
    \author{Matthias Bentert \and Fedor V. Fomin \and Fanny Hauser \and Saket Saurabh}
\fi

\usepackage{mathdots} 
\usepackage{amsthm}
\usepackage{amsmath}
\usepackage{amssymb}
\usepackage{dsfont}
\usepackage{graphicx}
\usepackage{nicefrac}
\usepackage{tikz}
\usetikzlibrary{calc,math,patterns,shapes,backgrounds,patterns.meta}
\usepackage{todonotes} 
\setuptodonotes{inline}
\usepackage{tabularx}
\usepackage[boxed,noline,noend,ruled,linesnumbered]{algorithm2e}
\usepackage{xintexpr}
\usepackage{boxedminipage}
\usepackage{framed}
\usepackage{tcolorbox}
\usepackage{etoolbox}
\usepackage{xifthen}
\usepackage{mathtools}

\title{The Parameterized Complexity Landscape of Two-Sets Cut-Uncut}
\date{}    

\newlength{\RoundedBoxWidth}
\newsavebox{\GrayRoundedBox}
\newenvironment{GrayBox}[1]%
   {\setlength{\RoundedBoxWidth}{.93\textwidth}
    \def\boxheading{#1}
    \begin{lrbox}{\GrayRoundedBox}
       \begin{minipage}{\RoundedBoxWidth}}%
   {   \end{minipage}
    \end{lrbox}
    \begin{center}
    \begin{tikzpicture}%
       \node(Text)[draw=black!20,fill=white,rounded corners,%
             inner sep=2ex,text width=\RoundedBoxWidth]%
             {\usebox{\GrayRoundedBox}};
        \coordinate(x) at (current bounding box.north west);
        \node [draw=white,rectangle,inner sep=3pt,anchor=north west,fill=white] 
        at ($(x)+(6pt,.75em)$) {\boxheading};
    \end{tikzpicture}
    \end{center}}
    
\newenvironment{defproblemx}[2][]{\noindent\ignorespaces%
                                \FrameSep=6pt%
                                \parindent=0pt%
                \vspace*{-1.5em}
                \ifthenelse{\isempty{#1}}{%
                  \begin{GrayBox}{\textsc{#2}}%
                }{%
                  \begin{GrayBox}{\textsc{#2}  parameterized by~{#1}}%
                }
                \begin{tabular*}{\textwidth}{@{\hspace{.1em}} >{\itshape} p{1.2cm} p{0.86\textwidth} @{}}%
            }{
                \end{tabular*}%
                \end{GrayBox}%
                \ignorespacesafterend
            }  

\newcommand{\defproblem}[3]{
  \begin{defproblemx}{#1}
    Input:  & #2 \\
    Question: & #3
  \end{defproblemx}
}%

\definecolor{pgreen}{HTML}{81B622}
\definecolor{pyellow}{rgb}{1.0, 0.88, 0.21}
\definecolor{porange}{rgb}{1.0, 0.66, 0.07}
\definecolor{pred}{HTML}{ff4d00}

\definecolor{fred}{HTML}{ba181b}
\definecolor{fblue}{HTML}{30638e}
\colorlet{fblue2}{fblue!15}
\colorlet{fred2}{fred!15}

\newcommand{\tworows}[2]{\begin{tabular}{c}{#1}\\{#2}\end{tabular}}
\newcommand{\distto}[1]{\tworows{Distance to}{#1}}
\tikzstyle{para}=[rectangle,draw=black,minimum height=.8cm,fill=gray!10,rounded corners=1mm, on grid]

\tikzset{inner sep=2.5pt}
\DeclareMathOperator{\cut}{cut}

\DeclareMathOperator{\poly}{poly}
\DeclareMathOperator{\NP}{NP}
\DeclareMathOperator{\coNP}{coNP}
\DeclareMathOperator{\join}{\Cup}
\newcommand{\cu}{\textsc{Two-Sets Cut-Uncut}}
\newcommand{\dcs}{\textsc{2-Disjoint Connected Subgraphs}}
\newcommand{\maxcu}{\textsc{Largest Bond}}
\newcommand{\ppoly}{\ensuremath{\NP \subseteq \coNP$\slash$\poly}}
\newcommand{\nppoly}{\ensuremath{\NP \not\subseteq \coNP$\slash$\poly}}
\newcommand{\Nr}[1]{\ensuremath{N^r_{#1}}}
\newcommand{\Nb}[1]{\ensuremath{N^b_{#1}}}

\usepackage{etoolbox}

\newcommand{\appmark}{$\star$}
\newcommand{\appendixText}{}
\newcommand{\toappendix}[1]{%
    \ifconf
        \gappto{\appendixText}{
            #1
        }
    \else #1 \fi
}

\newcommand{\appendixproof}[3]{%
    \ifconf
        \gappto{\appendixText}{
            \paragraph*{Proof of \cref{#1}}
                #2
                #3
            }
    \else #3 \fi
}

\ifconf
    
\else
    \usepackage{thm-restate}
    \usepackage[colorlinks=true, allcolors=blue]{hyperref}
    \usepackage[capitalize,noabbrev]{cleveref}
    
    \newtheorem{theorem}{Theorem}
    \newtheorem{definition}{Definition}

    \newtheorem{observation}{Observation}
    \newtheorem{corollary}{Corollary}
\fi

\begin{document}

\maketitle

\begin{abstract}
    In \cu, we are given an undirected graph~$G=(V,E)$ and two terminal sets~$S$ and~$T$.
    The task is to find a minimum cut $C$ in~$G$ (if there is any) separating~$S$ from~$T$ under the following ``uncut'' condition.
    In the graph~$(V,E\setminus C)$, the terminals in each terminal set remain in the same connected component.
    In spite of the superficial similarity to the classic problem \textsc{Minimum~$s$-$t$-Cut}, \cu{} is computationally challenging. In particular, even deciding whether such a cut of \emph{any size} exists, is already NP-complete.
    We initiate a systematic study of \cu{} within the context of parameterized complexity.
    By leveraging known relations between many well-studied graph parameters, we characterize the structural properties of input graphs that allow for polynomial kernels, fixed-parameter tractability (FPT), and slicewise polynomial algorithms (XP).
    Our main contribution is the near-complete establishment of the complexity of these algorithmic properties within the described hierarchy of graph parameters.
    
    On a technical level, our main results are fixed-parameter tractability for the (vertex-deletion) distance to cographs and an OR-cross composition excluding polynomial kernels for the vertex cover number of the input graph (under the standard complexity assumption \nppoly).
\end{abstract}

\clearpage

\section{Introduction}
We study \cu, a natural optimization variant of \dcs.
In \dcs, we are given an undirected graph~$G$ and two disjoint sets~$S,T$ of vertices.
The question is whether there are two disjoint sets~$R,B$ of vertices such that~$S \subseteq R$, $T \subseteq B$, and both~$G[R]$ and~$G[B]$ (the graphs induced by~$R$ and~$B$, respectively) are connected.
\dcs{} is the special case of \textsc{Disjoint Connected Subgraphs} with two sets of terminals (a problem that played a crucial role in the graph-minors project by Robertson and Seymour~\cite{RS95}).
Consequently, \mbox{\dcs} has received considerable research attention, particularly from the graph-algorithms and the computational-geometry communities~\cite{CPPW14,GKLS12,HPW09,KMPSL22,PR11,TV13}.
In \cu, we not only want to decide whether there are disjoint connected sets containing terminal sets~$S$ and~$T$, respectively, but also minimize the size of the corresponding cut (if it exists).
Formally, \cu{} is defined as follows.
Therein, an~$S$-$T$-cut~$(R,B)$ is a partition of the set of vertices into~$R$ and~$B$ with~$S \subseteq R$ and~$T \subseteq B$.
The set~$\cut_G(R)$ contains all edges in~$G$ with exactly one endpoint in~$R$ (and the other in~$B$).

\defproblem{\cu}
{A connected undirected graph~$G=(V,E)$, two sets~$S,T \subseteq V$, and an integer~$\ell$.}
{Is there an~$S$-$T$-cut~$(R,B)$ of~$G$ with~$|\cut_G(R)| \leq \ell$ such that the vertices of~$S$ are in the same connected component of~$G[R]$ and the vertices of~$T$ are in the same connected component of~$G[B]$?}
We mention in passing that we assume the input graph to be connected as it becomes trivial if there are at least two connected components containing terminal vertices (vertices in~$S \cup T$) and if all terminals belong to one connected component, then we can discard all other connected components.
When~$G$ is connected, it is also easy to see that any optimal solution for \cu{} cuts the graph in exactly two connected components as any connected component not containing any terminal can be merged with any other connected component reducing the size of the cut in the process.
Finally, if~$S \cap T \neq \emptyset$, then the instance is a trivial no-instance.

\subparagraph*{Related work.}
\dcs{} was intensively studied and its complexity is quite well understood.
Gray et al.\;\cite{GKLS12} showed that \dcs{} is NP-hard on planar graphs and van 't Hoft et al.\;\cite{HPW09} showed NP-hardness even if~$|S| = 2$ and on~$P_5$-free split graphs.
Note that the problem becomes trivial if~$|S| = 1$.
Since split graphs are chordal, their results also show that \dcs{} is NP-hard on split graphs.
They also complemented the NP-hardness on~\mbox{$P_5$-free} graphs by providing a polynomial-time algorithm for~$P_4$-free graphs (also known as cographs).
Kern et al.\;\cite{KMPSL22} generalized this result by showing that for each graph~$H$, \dcs{} is polynomial-time solvable on~$H$-free graphs if and only if~$H$ is a subgraph of a~$P_4$ together with any number of isolated vertices (and otherwise NP-hard). 
Cygan et al.\;\cite{CPPW14} studied the parameterized complexity with respect to the number~$k = n - |S \cup T|$ of non-terminal vertices in the graph.
They showed that \dcs{} cannot be solved in~$O^*((2-\varepsilon)^k)$~time for any~$\varepsilon > 0$ unless the strong exponential time hypothesis fails.
Moreover, they showed that it does not admit a polynomial kernel for this parameter unless~\ppoly.
As \dcs{} is the special case of \cu{} where~$\ell = m$, all of the above hardness results 

The problem \cu{} was introduced by Bentert et al.\;\cite{BDFGK24} who showed that the problem is W[1]-hard parameterized by~$|T|$ even if~$|S|=1$ in general graphs but fixed-parameter tractable when parameterized by~$|S \cup T|$ in planar graphs.
They also showed fixed-parameter tractability on planar graphs, when parameterized by the minimum size of a set of faces in any planar embedding such that each terminal is incident to one of the faces in the set.
Moreover, \cu{} is a special case of \textsc{Mixed Multiway Cut-Uncut}.
In this problem, one is not restricted to two sets of terminals and one is given two integers~$k$ and~$\ell$ as input.
The question is whether one can delete at most~$k$ vertices and at most~$\ell$ edges to separate all terminals in different sets while maintaining connectivity within each terminal set.
Rai et al.~\cite{RRS16} showed that this problem is fixed-parameter tractable when parameterized by~$k+\ell$ which immediately implies that \cu{} is fixed-parameter tractable when parameterized by~$\ell$.

Another related problem is called \maxcu.
Here, we are looking for a largest cut that cuts a connected graph into exactly two connected components (and no terminal sets are given).
Duarte et al.\;\cite{D+21} showed that the problem is NP-hard on bipartite split graphs, fixed-parameter tractable when parameterized by treewidth, does not admit a polynomial kernel when parameterized by the solution size, can be solved in~$f(k) n^{O(k)}$ time, where~$k$ is the clique-width of the input graph, but not in~$f(k) n^{o(k)}$ time unless the exponential time hypothesis fails.
In particular, this also excludes fixed-parameter tractability.

Last but not least, \cu{} is closely related to \textsc{Network Diversion}, which has
been studied extensively by the operations-research and networks communities~\mbox{\cite{CWN13,Cur01,Erk02,Kal15,LCP19}}.
In this problem, we are given an undirected graph~$G$, two terminal vertices~$s$ and~$t$, an edge~$b = \{u,v\}$, and an integer~$\ell$.
The question is whether it is possible to delete at most~$\ell$~edges such that the edge~$b$ will become a bridge with~$s$ on one side and~$t$ on the other.
Equivalently, is there a minimal~$s$-$t$-cut of size at most~$\ell + 1$ containing~$b$.
While this problem seems very similar to the classic \textsc{Minimum $s$-$t$-Cut}, the complexity status of this problem (polynomial-time solvable or NP-hard) is widely open.
This problem is a special case of \cu{} where~$|S| = |T| = 2$ as there are only two cases.
Either~$s$ is in the same component as~$u$ or~$s$ is in the same component as~$v$.
These two cases correspond to instances of \cu{} with~$S = \{s, u\}$ and~$T = \{t, v\}$ and~$S=\{s,v\}$ and~$T = \{t,u\}$, respectively.

\subparagraph{Our contribution.}
We provide an almost complete tetrachotomy for \cu{} distinguishing between parameters that allow for polynomial kernels, fixed-parameter tractability, or slicewise polynomial (XP-time) algorithms.
Our results are summarized in \cref{fig:results}.
The rest of this work is organized as follows.
In \cref{sec:prelim}, we introduce concepts and notation used throughout the paper.
In \cref{sec:fpt}, we present fixed-parameter tractable and slicewise polynomial (XP-time) algorithms.
In \cref{sec:hard}, we exclude the possibility for further fixed-parameter tractable or XP-time algorithms by presenting W[1]-hardness and para-NP-hardness results, respectively.
\cref{sec:kernel} is devoted to both positive and negative results regarding the existence of polynomial kernels and we conclude with \cref{sec:concl}.

\begin{figure}[t!]
\centering
\begin{tikzpicture}[node distance=1.75*0.45cm and 3*0.45cm, every node/.style={scale=0.6}]
\node[para,fill=pyellow] (vc) {Minimum Vertex Cover};
\node[para, xshift=3.4cm,fill=pgreen] (ml) [right=of vc] {Max Leaf \#};
\node[para, xshift=-3.5cm,fill=pyellow] (dc) [left=of vc] {Distance to Clique};
\node[para,xshift=7.8cm,fill=pyellow] (ss) [right=of vc] {Solution Size};
\node[para,xshift=-7cm,pattern=north east lines, pattern color=pred] (terminals) [left=of vc] {\# of Terminals};
\node[para, xshift=-10mm,fill=porange] (mcc) [below left=of dc] {\tworows{Minimum}{Clique Cover}}
edge (dc);
\node[para,fill=pyellow] (dcc) [below= of dc] {\distto{Co-Cluster}}
edge (dc)
edge (vc);
\node[para,fill=pyellow] (dcl) [below left= of vc] {\distto{Cluster}}
edge (dc)
edge (vc);
\node[para, xshift=5mm,fill=pyellow] (ddp) [below=of vc] {\distto{disjoint Paths}}
edge (vc)
edge (ml);
\node[para,fill=pgreen] (fes) [below =of ml] {\tworows{Feedback}{Edge Set}}
edge (ml);
\node[para,fill=pyellow] (bw) [below right=of ml] {Bandwidth}
edge (ml);

\node[para] (is) [below=of mcc,fill=porange] {\tworows{Maximum}{Independent Set}}
edge (mcc);
\node[para,fill=pyellow] (dcg) [below= of dcc] {\distto{Cograph}}
edge (dcc)
edge (dcl);
\node[para] (dig) [below= of dcl] {\distto{Interval}}
edge (dcl)
edge (ddp);
\node[para,fill=pyellow] (fvs) [below= of ddp] {\tworows{Feedback}{Vertex Set}}
edge (ddp)
edge (fes);
\node[para, xshift=10mm, yshift=1mm,fill=pyellow] (td) [right=of fvs] {Treedepth}
edge [bend right=20] (vc);
\node[para,fill=pred] (mxd) [below= of bw] {\tworows{Maximum}{Degree}}
edge (bw);
\node[para, xshift=1mm,fill=pred] (bsw) [below right= of bw] {\tworows{Bisection}{Width}}
edge (bw);

\node[para,fill=pred] (ds) [below=of is] {\tworows{Minimum}{Dominating Set}}
edge (is);
\node[para, xshift= 15mm, yshift=-35mm,fill=pred] (dch) [below= of dcg] {\distto{Chordal}}
edge (dig)
edge (fvs);
\node[para, yshift=-71mm, xshift=-10mm,fill=pred] (dbp) [below left= of fvs] {\distto{Bipartite}}
edge (fvs);
\node[para, xshift=10mm, yshift=-5mm,fill=pyellow] (dop) [below= of fvs] {\distto{Outerplanar}}
edge (fvs);
\node[para, yshift=3mm,fill=pyellow] (pw) [below= of td] {Pathwidth}
edge (ddp)
edge (td)
edge (bw);
\node[para,fill=pred] (hid) [below= of mxd] {$h$-index}
edge [bend right=15] (ddp)
edge (mxd);
\node[para, yshift=-5mm,fill=pred] (gen) [left= of hid] {Genus}
edge (fes);

\node[para,fill=pred] (mxdia) [below=of ds] {\tworows{Max Diameter}{of Components}}
edge (ds)
edge[bend right=5] (td)
edge[bend right=25] (dcg);
\node[para, yshift=-110mm, xshift=-10mm,fill=pred] (dpf) [below= of dcc] {\distto{Perfect}}
edge (dch)
edge (dcg)
edge (dbp);
\node[para, yshift=-2mm,fill=pyellow] (tw) [below= of pw] {Treewidth}
edge (dop)
edge (pw);
\node[para, xshift=5mm, yshift=-15mm,pattern=north east lines, pattern color=porange] (clw) [below left= of tw] {Clique-width}
edge (dcg)
edge (tw);
\node[para, yshift=-15mm,fill=pred] (acn) [below = of gen] {\tworows{Acyclic}{Chromatic \#}}
edge (hid)
edge (gen)
edge (tw);
xtxtxt\node[para, yshift=-2.5mm,fill=pred] (dpl) [below = of dop] {\distto{Planar}}
edge (dop)
edge (acn);
\node[para,fill=pred] (avgdist) [below=of mxdia] {\tworows{Average}{Distance}}
edge (mxdia);
\node[para,fill=pred] (deg) [below= of acn] {Degeneracy}
edge (acn);
\node[para, yshift=-12mm, xshift=-25mm,fill=pred] (box) [below left=of acn] {Boxicity}
edge (dig)
edge (acn);
\node[para,fill=pred] (cn) [below =of dbp] {Chromatic \#}
edge (deg)
edge (dbp);
\node[para, yshift= -13mm, xshift= 30mm,fill=pred] (cho) [right= of cn] {Chordality}
edge (box)
edge (dch)
edge (cn)
edge (dcg);
\node[para,fill=pred] (avd) [below=of deg] {\tworows{Average}{Degree}}
edge (deg);

\node[para,fill=pred] (mnd) [below=of avd] {\tworows{Minimum}{Degree}}
edge (avd);

\node[para, yshift=-97mm, xshift=-15mm,fill=pred] (con) [below=of bsw] {\distto{Disconnected}}
edge (mnd)
edge (bsw);
\node[para, xshift=9mm,fill=pred] (don) [right=of cho] {Domatic \#}
edge (mnd);
\node[para,fill=pred] (mc) [below=of cn] {\tworows{Maximum}{Clique}}
edge (cn);
\node[para,fill=pred] (gi) [below=of avgdist] {Girth}
edge (avgdist);
\end{tikzpicture}
\caption{
    Overview of our results.
    An edge between two parameters~$\alpha$ and~$\beta$, where~$\alpha$ is above~$\beta$, indicates that in any instance, the value of~$\beta$ is upper-bounded by a function only depending on the value of~$\alpha$.
    Any hardness result for~$\alpha$ immediately implies the same hardness result for~$\beta$ and any positive result for~$\beta$ immediately implies the same positive result for~$\alpha$ (where we additionally require that the dependency is polynomial if we show or exclude a polynomial kernel).
    Green boxes indicate the existence of polynomial kernels, yellow boxes show that the parameter admits fixed-parameter tractability but no polynomial kernel, an orange box indicates polynomial-time algorithms for constant parameter values (XP) but no fixed-parameter tractability, and a red box shows that  the parameter is NP-hard for some constant parameter value.
    We mention that the status of \cu{} parameterized by distance to interval graphs, number of terminals (XP/para-NP-hard), and clique-width (fixed-parameter tractable/W[1]-hard) remain open.
}
\label{fig:results}
\end{figure}
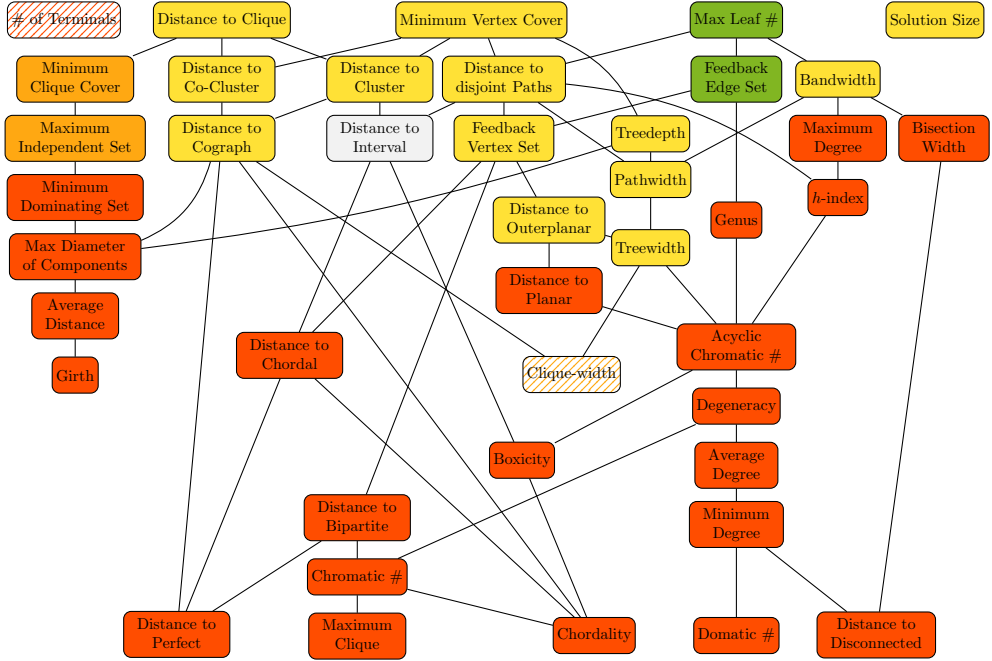

\section{Preliminaries}
\label{sec:prelim}

For a positive integer~$n$, let~${[n] = \{1, 2, \dots, n\}}$.
We use standard graph-theoretic terminology and all graphs in this work are undirected.
In particular, for an undirected graph~${G = (V, E)}$ we set~$n = |V|$ and~$m = |E|$.
For a subset~$V' \subseteq V$ of the vertices, we use~$G[V']$ to denote the subgraph of~$G$ induced by~$V'$ and denote by~$G - V'$ the subgraph~$G[V \setminus V']$. 
Moreover, for an edge set~$E' \subseteq E$, we denote by~$G - E' = (V, E \setminus E')$ the graph resulting from deleting the edges in~$E'$ from~$G$.
The \emph{degree}~$\deg_G(v)$ of $v$ is the number of vertices adjacent to~$v$ in~$G$.
A \emph{path} $P = (v_1,v_2,\ldots,v_{\ell})$ \emph{on~$\ell$~vertices} is a graph with vertex 
set~$\{v_1, v_2, \ldots, v_{\ell}\}$ and edge 
set~${\{ \{v_i, v_{i+1}\} \mid i \in [\ell-1]\}}$.
The vertices~$v_1$ and~$v_{\ell}$ are called \emph{endpoints}.
The \emph{length} of a path is its number of edges.
A \emph{connected component} in a graph is a maximal set~$V'$ of vertices such that for each pair~$u,v \in V'$, there is a path in the graph with endpoints~$u$ and~$v$.
A \emph{cut} in a graph is the set of edges between any partition of the vertices of a graph into two disjoint subsets.
A \emph{separation} in a graph~$G=(V,E)$ is a pair~$(X,Y)$ of sets of vertices with~$X \cup Y = V$ and no edges between~$X \setminus Y$ and~$Y \setminus X$.
The size of the separation is~$X \cap Y$.
The \emph{disjoint union} of two graphs~$G_1=(V_1,E_1)$ and~$G_2=(V_2,E_2)$ results in the graph~$(V_1 \cup V_2,E_1\cup E_2)$.
The \emph{join} of~$G_1$ and~$G_2$ results in the graph~$(V_1 \cup V_2,E_1\cup E_2 \cup \{\{u,v\}\mid u\in V_1 \land v \in V_2\})$, that is, we first take the disjoint union and then add all possible edges between the two graphs.
We refer to the Bachelor's thesis of Schröder \cite{Sch19} for an overview over how the different parameters are related to one another.

To streamline some of our arguments, we use the following natural reinterpretation of \cu.
The task is to color each vertex in the graph red or blue such that all vertices in~$S$ are red, all vertices in~$T$ are blue, the graphs induced by the set of all red vertices (and all blue vertices, respectively) are connected, and there are at most~$\ell$ edges with a red and a blue endpoint.
We call such edges multicolored.
We often keep sets~$R$ and~$B$ of red and blue vertices, respectively, and we use the notation~$\Nr{X}(v)$ and~$\Nb{X}(v)$ to denote all red and blue neighbors of~$v$ in a set~$X$ of vertices, respectively.

\subparagraph{Parameterized complexity.}
A \emph{parameterized problem} is a set of instances~$(I, k)$ where~$I \in \Sigma^*$ is a problem instance from some finite alphabet~$\Sigma$ and the integer~$k$ is the \emph{parameter}.
A parameterized problem~$L$ is \emph{fixed-parameter tractable} if $(I, k) \in L$ can be decided in~${f(k) \cdot |I|^{O(1)}}$~time, where~$f$ is a computable function only depending on~$k$.
We call~$(I, k)$ a \emph{yes-instance} (of~$L$) if~$(I, k) \in L$.
The class XP contains all parameterized problems which can be decided in polynomial time if the parameter~$k$ is constant, that is, in $f(k)\cdot |I|^{g(k)}$ time for computable functions~$f$ and~$g$.
It follows from the definition that each fixed-parameter tractable problem is contained in XP.
To show that a problem is not contained in XP, one can show that the problem remains NP-hard for some constant parameter value.
To show that a parameterized problem $L$ is presumably not fixed-parameter tractable, one may use a \emph{parameterized reduction} from a W[1]-hard problem to~$L$~\cite{DF13}.
A parameterized reduction from a parameterized problem $L$ to another parameterized problem~$L'$ is an algorithm that, given an instance~$(I, k)$ of~$L$, computes an instance~$(I', k')$ of $L'$ in~$f(k) \cdot |I|^{O(1)}$~time such that $(I, k)$ is a yes-instance if and only if $(I', k')$ is a yes-instance and~$k' \leq g(k)$ for two computable functions~$f$ and~$g$.
A \emph{kernelization} is an algorithm that, given an instance~$(I, k)$ of a parameterized problem~$L$, computes in~$|I|^{O(1)}$ time an instance~$(I', k')$ of~$L$ (the \emph{kernel}) such that~$(I, k)$ is a yes-instance if and only if~$(I', k')$ is a yes-instance and~$|I'|+k' \le f(k)$ for some computable function~$f$ only depending on~$k$.
We say that~$f$ measures the \emph{size} of the kernel.
If~$f$ is a polynomial, then we say that~$P$ admits a \emph{polynomial kernel}.
A problem is fixed-parameter tractable if and only if it admits a kernel of any size.
Assuming \nppoly, one can show that certain parameterized problems do not admit a polynomial kernel.
This can for example be done via OR-cross-compositions.
For the definition of OR-cross-compositions, we first need the following.
Given an NP-hard problem~$L$, an equivalence relation~$\mathcal R$ on the instances of~$L$ is a \emph{polynomial equivalence relation} if one can decide for any two instances in polynomial time whether they belong to the same equivalence class, and for any finite set~$S$ of instances, $\mathcal R$ partitions the set into at most~$(\max_{I \in S} |I|)^{O(1)}$ equivalence classes.

\begin{definition}[OR-cross-composition \cite{BJK14}]
	\label{def:or-cross-composition}
	Given an NP-hard problem~$Q$, a para\-meterized problem~$L$, and a polynomial equivalence relation~$\mathcal R$ on the instances of~$Q$,
	an OR-cross-composition of~$Q$ into~$L$ (with respect to~$\mathcal R$) is an algorithm that takes~$t$ instances~$I_1, I_2, \dots, I_t$ of~$Q$ belonging to the same equivalence class of~$\mathcal{R}$ and constructs in time polynomial in~$\sum_{i=1}^t |I_i|$ an instance~$(I, k)$ of~$L$ such that~$k$~is polynomially upper-bounded by~$\max_{i \in [t]} |I_i| + \log(t)$ and~$(I, k)$ is a yes-instance of~$L$ if and only if there exists an~$i \in [t]$ such that~$I_i$ is a yes-instance of~$Q$.
\end{definition}
If a parameterized problem admits an OR-cross-composition, then it does not admit a polynomial kernel unless \ppoly~\cite{BJK14}.

\toappendix{
\subparagraph*{Graph parameters and classes.}
We give an overview of the different graph parameters and graph classes used throughout the paper.
To this end, let~$G=(V,E)$ be a graph.
The \emph{maximum degree} of~$G$ is the largest degree of any vertex in~$V$.
The \emph{distance to~$\Pi$} for some graph class~$\Pi$ is the minimum number of vertices needed to be removed from~$G$ such that it becomes a graph in~$\Pi$.
A \emph{cograph} is a graph without induced paths of length three.
Equivalently, a cograph is a graph that can be represented by a cotree.
A cotree is a rooted binary tree in which the leaves correspond to the vertices in the cograph and the internal nodes correspond to taking the disjoint union or the join of the cographs corresponding to the two children.
A join of two graphs~$G_1=(V_1,E_1)$ and~$G_2=(V_2,E_2)$ results in the graph~$(V_1 \cup V_2,E_1\cup E_2 \cup \{\{u,v\}\mid u\in V_1 \land v \in V_2\})$, that is, we first take the disjoint union and then add all possible edges between the two graphs.
An \emph{interval graph} is a graph where each vertex can be represented by an interval of real numbers such that two vertices are adjacent if and only if their respective intervals overlap.
An \emph{independent set} in a graph is a set of pairwise non-adjacent vertices.
The vertex set of a \emph{bipartite graph} can be partitioned into two independent sets.
The \emph{vertex cover number} of~$G$ is the distance to an independent set.
A \emph{clique} in a graph is a set of pairwise adjacent vertices.
The \emph{minimum clique cover} of~$G$ is the minimum number of cliques needed to partition~$V$.
A \emph{dominating set} in a graph is a set of vertices such that each vertex not contained in the set has at least one neighbor in the set.
A \emph{tree decomposition} of~$G$ is a tree~$T$ with nodes~$X_1,X_2,\ldots,X_p$, where each~$X_i$ is a subset of~$V$ such that each vertex in~$V$ and both endpoints of each edge in~$E$ are contained in some~$X_i$ and all~$X_i$ that contain some vertex~$v$ form one connected component in~$T$.
The width of a tree decomposition is the size of its largest set~$X_i$ minus one and the \emph{treewidth} of~$G$ is the minimum width of any tree decomposition of~$G$.
The \emph{clique-width} of~$G$ is defined as the minimum number of labels needed to construct~$G$ using the following operations in which~$i$ and~$j$ are some arbitrary labels:
Creating a single vertex with label~$i$, the disjoint union of two graphs, adding all possible edges between the vertices of label~$i$ and the vertices of label~$j$, and changing the label of all vertices of label~$i$ to label~$j$.
The \emph{feedback edge number} of~$G$ is the minimum size of a set~$F$ of edges such that~$G-F$ is a forest.
Given an injective function~$f$ that maps the vertices in~$V$ to distinct integers, the bandwidth cost of~$f$ for~$G$ is defined as~$\max_{\{u,v\}\in E}|f(u) - f(v)|$.
The \emph{bandwidth} of~$G$ is the minimum bandwidth cost for~$G$ over all possible injective functions.
The \emph{bisection width} of~$G$ is the minimum size of a set~$F$ of edges such that the vertices of~$G-F$ can be partitioned into two parts of equal size (or with a difference of one in case~$n$ is odd) with no edges between the two parts.
}

\section{Parameterized Algorithms}
In this section, we present some parameterized algorithms for \cu.
We start with the distance to cographs. 
\label{sec:fpt}
\begin{theorem}
    \cu{} parameterized by the distance~$k$ to cographs can be solved in~$k^{O(k)} n^3$~time.
\end{theorem}

\begin{proof}
Let $(G = (V,E),S,T, \ell)$ be an instance of \cu{}. We transform~$G$ into a weighted graph by assigning a weight of one to each edge in~$G$.
We denote the weight of an edge~$e$ by~$w(e)$.
We first compute in~$O(3.303^k(m + n))$~time a set~$X \subseteq V$ of size at most~$k$ such that~$G'=(V',E') = G-X$ is a cograph \cite{NG12}.
A cotree~$\mathcal{T}$ of~$G'$ is a rooted binary tree where each vertex of~$G'$ corresponds to a leaf node of~$\mathcal{T}$ and an inner node~$t$ of~$\mathcal{T}$ either represents taking the disjoint union or the join of the cographs corresponding to the two children of $t$.
Each cograph has a cotree which can be computed in linear time.
For each node $t$ of $\mathcal{T}$, let~$\mathcal{T}_t$ be the subtree of~$\mathcal{T}$ rooted in~$t$, let~$G_t = (V_t,E_t)$ be the graph represented by~$\mathcal{T}_t$, and let~$n_t = |V_t|$.
For technical reasons, we want to assume that~$X$ contains at least one red and one blue vertex in an optimal solution.
Hence, if~$X$ does not contain a terminal from~$S$ already, then we add an arbitrary vertex from~$S$ to~$X$.
We do the same for~$T$.

We begin by guessing\footnote{Whenever we pretend to guess something, we actually iterate over all possibilities and consider for the presentation/proof an iteration leading to an optimal solution.} the coloring of $X$.
Then, we remove all edges between a red and a blue vertex in $X$. Let $\ell'$ be the number of edges that we removed.
We then contract each component of~$X$ into a single vertex. If the resulting (multi)graph has $j$ edges with the same endpoints (with one endpoint in $X$ and one in $V'$), then we remove all but one of the edges and set its weight to $j$.

We compute a coloring for the remaining vertices of $G$ by computing optimal solutions for partial instances for each node of $\mathcal{T}$ via dynamic programming.
A partial instance is a tuple $(t, r, C_r, C_b)$ where $t$ is a node of $\mathcal{T}$, $r$ is an integer at most~$n_t$ and~$C_r, C_b$ are partitions of subsets of $X$ including $\emptyset$.
For each partial instance, we want to store in a table~$D$ the minimum number of edges between blue and red vertices in $G[V_t \cup X]$ after all vertices in~$V_t$ are colored and exactly~$r$~vertices are colored red.
Additionally, we require that for each set~$C \in C_r$, there exists a connected component with vertex set~$Z$ in~$G[(V_t \cup X) \cap R]$ such that~$Z \cap X \cap R = C$.
Analogously for each set $C \in C_b$, there has to exist a connected component with vertex set $Z$ in $G[(V_t \cup X) \cap B]$ such that $Z \cap X \cap B = C$.
We will use these sets to store whether the red vertices are connected in $G[V_t \cup X]$, while also storing which vertices of $X\cap R$ are connected to the same connected component in~$G_t[R]$.
We use this information to ensure that in a coloring for the entire graph~$G$, the graph $G[R]$ is connected.
We do the same for the blue vertices with the set~$C_b$.

In the end, the optimal solution (minus the~$\ell'$ edges we already removed between vertices in~$X$) will be stored in~$D[w,r,C_r,C_b]$ for some value of~$r$, where~$w$ is the root of~$\mathcal{T}$ and~$C_r = \emptyset$ if~$r=0$ and~$C_r=\{X\cap R\}$, otherwise.
Similarly,~$C_b = \emptyset$ if~$r=n_w$ and~$C_b = \{X \cap B\}$, otherwise.
Note that this corresponds to a solution where all vertices of either color form a single connected component as all vertices are connected to all vertices of the same color in~$X$ (which is at least one vertex as constructed above).

Before we present the algorithm, we first define two operations on sets of sets.
For two sets~$A,B$ of sets of vertices, if there exists~$a \in A$ and~$b \in B$ with $a\cap b \neq \emptyset$, then we recursively define~$A \uplus B$ as~$((A\setminus a) \uplus (B\setminus b))\uplus \{a \cup b\}$ and as~$A \cup B$ otherwise. 
This operation can be seen as taking the union of the connected components of two subgraphs.
If there are two connected components (one in~$A$ and one in~$B$) which share a vertex, then both components are merged into one.
The components that do not share a vertex with any other component remain as they are. 
We also define~$A \join B$ as $A$ if~$B = \emptyset$, as~$B$ if~$A = \emptyset$ and as~$\{\bigcup_{X \in A \cup B} X\}$, otherwise.
With these definitions at hand, we compute the entries of~$D$ based on the type of node~$t$ as follows.

\subparagraph{Leaves.}

Let $t$ be a leaf node and let $a$ be the vertex of $G$ corresponding to $t$. We set 
\[
D[t,r,C_r, C_b] = 
\begin{cases}
\sum\limits_{v \in \Nb{X}(a)} w(\{a,v\}), &\text{ if } r=1, C_r = \{\Nr{X}(a)\}, C_b = \emptyset \text{ and } a \notin T\\
\sum\limits_{v \in \Nr{X}(a)} w(\{a,v\}), & \text{ if }  r=0,C_r= \emptyset, C_b = \{\Nb{X}(a)\} \text{ and } a \notin S\\
\infty,  & \text{ else.}
\end{cases}
\]

We show that the solutions for all partial instances are computed correctly.
It is easy to verify that whenever a table entry of~$D$ is set, then this corresponds to a valid solution for the partial instance.
So it remains to show that each optimal solution for any partial instance is considered.
Let~$(t,r,C_r,C_b)$ be a partial instance and consider an optimal solution for it.
Note that the partial instance only has a valid solution if~$r \in \{0,1\}$.
If~$r=1$, then the vertex~$a$ has to be colored red.
This is only possible if~$a \notin T$ since all vertices in~$T$ have to be colored blue.
The red vertices of~$X$ which are adjacent to vertices of~$G_t[R]$ are the vertices in~$\Nr{X}(a)$. If~$a$ is not adjacent to any red vertices of~$X$, then~$C_r = \{\emptyset\}$ or there is no solution to the partial instance.
Since~$G_t$ cannot have any blue vertices, $C_b$ has to be the empty set (and not the set containing the empty set).
Note that all edges between~$a$ and blue neighbors of~$a$ in~$X$ are multicolored and this is precisely what we computed above.
The argument for~$r=0$ (coloring~$a$ blue) is symmetric.

\subparagraph{Disjoint union.}
Let $t$ be a disjoint-union node with children $t_1$ and $t_2$.
We set
\[
D[t,r, C_r, C_b] = \smashoperator{ \min_{\substack{r_1 \\ C_r = C_{r_1} \uplus C_{r_2}\\ C_b = C_{b_1} \uplus C_{b_2}}}} (D[t_1, r_1, C_{r_1}, C_{b_1}] + D[t_2, r-r_1, C_{r_2}, C_{b_2}]).
\]

We again show that each optimal solution for any partial instance is considered.
To this end, we assume that the table entries for the children~$t_1$ and~$t_2$ are computed correctly.
Let~$(t,r,C_r,C_b)$ be a partial instance and consider an optimal solution for it.
Since the disjoint union of two graphs does not create any edges, it holds that all multicolored edges in the solution are contained in~$G[V_{t_1} \cup X]$ and~$G[V_{t_2} \cup X]$.
This gives a partitioning of the multicolored edges in the solution as there are no edges between vertices in~$X$.
Let~$r_1$ be the number of red vertices in~$V_{t_1}$ in the solution and consider the connected components in~$G[(V_{t_1} \cup X) \cap R]$.
Denote the respective partitioning by~$C_{r_1}$ and do the same for~$C_{b_1}, C_{r_2},$ and~$C_{b_2}$.
Note that~$C_r$ corresponds to the union of the connected components corresponding to~$C_{r_1}$ and~$C_{r_2}$ (and the same for~$C_b$).
Hence, it holds that~$C_r = C_{r_1} \uplus C_{r_2}$ and~$C_b = C_{b_1} \uplus C_{b_2}$.
It now also holds that the size of the considered solution is~$D[t_1,r_1,C_{r_1},C_{b_1}] + D[t_2,r-r_1,C_{r_2},C_{b_2}]$, which is precisely what we computed.

\subparagraph{Join.}
Let $t$ be a join node with children $t_1, t_2$.
We set (with~$r_2 = r-r_1$)
\begin{align*}
    D[t,r, C_r, C_b] = \smashoperator{\min_{\substack{r_1 \\ C_r = C_{r_1} \join C_{r_2}\\ C_b = C_{b_1} \join C_{b_2}}}} \{D[t_1, r_1, C_{r_1}, C_{b_1}] + D[t_2, r_2, C_{r_2}, C_{b_2}]  + r_1(n_2-r_2)+r_2(n_1-r_1)\}.
\end{align*}

We show a final time that each optimal solution for any partial instance is considered when the table entries for the children~$t_1$ and~$t_2$ are computed correctly.
Let~$(t,r,C_r,C_b)$ be a partial instance and consider an optimal solution for it.
Similar to the disjoint union, we can partition all multicolored edges of the solution.
In this case, we partition the multicolored edges into four parts, the edges in~$G[V_{t_1} \cup X]$,~$G[V_{t_2} \cup X]$, the newly created edges between red vertices in~$V_{t_1}$ and blue vertices in~$V_{t_2}$ and similarly between blue vertices in~$V_{t_1}$ and red vertices in~$V_{t_2}$.
Let~$r_1$ be the number of red vertices in~$V_{t_1}$ in the solution and again consider the connected components in~$G[(V_{t_1} \cup X) \cap R]$.
Denote the respective partitioning by~$C_{r_1}$ and do the same for~$C_{b_1}, C_{r_2},$ and~$C_{b_2}$.
If there is at least one red vertex in each of~$V_{t_1}$ and~$V_{t_2}$, then all red vertices of~$G_t$ will be connected in~$G_t[R]$. Thus the red vertices of~$X$ which where connected to any vertex of~$G_t[R]$ will be in the same set in~$C_r$.
If either~$V_{t_1}$ or $V_{t_2}$ do not contain any red vertices in the considered solution, then no edges between two red vertices are added and the connected components of $G_t[R]$ are the same as of $G_{t_1}[R] \cup G_{t_2}[R]$ (as at least one of the two sets is the empty set).
This is precisely what the $\join$ operator computes and the argument for the blue vertices is analogous. 
Hence, the optimal solution is~$D[t_1,r_1,C_{r_1},C_{b_1}] + D[t_2,r-r_1,C_{r_2},C_{b_2}] + r_1 (n_2 - (r-r_1)) + (r-r_1)(n_1 - r_1)$, which is what we compute in the dynamic program.

It remains to analyze the running time.
Note that both~$\uplus$ and~$\join$ can be computed in~$O(k^2)$~time.
The number of possible guesses for the coloring of~$X$ is at most~$2^{k}$.
The number of partial instances is at most~$O(n^2 (k+2)^{k+2})$ as there are at most~$n$ possibilities for~$t$ and~$r$ each and both~$C_r$ and~$C_b$ are partitions of subsets of $X$ which have size at most $k+2$ (as we added a terminal of~$S$ and of~$T$ to~$X$).
Computing each entry takes~$O((k+2)^{4(k+2)+2} \cdot n)$~time.
Thus, the overall runtime is in~$O(2^k \cdot (k+2)^{6k+14}\cdot n^3) \subseteq k^{O(k)} \cdot n^3$.
\end{proof}

We next show fixed-parameter tractability for the parameter treewidth.
We mention that the algorithm is a simple adaptation of a dynamic program for \maxcu{} \cite{D+21}.
In essence, each entry in the dynamic program over the tree decomposition (except for the leaves) are computed by combining the solutions for the children using a maximum over all combinations of solutions that color the vertices in the given bag in a certain way and ensure certain connectivity conditions.
By replacing the maximum with a minimum, we can solve a version of \cu{} without terminals (which is not a hard problem to solve).
However, we can also incorporate terminals by setting the value of all leaf nodes corresponding to a vertex~$v$ to infinity whenever~$v \in S$ and the given coloring of the bag colors~$v$ blue (or analogously~$v \in T$ and the coloring for~$v$ is red).
This yields the following.

\begin{observation}
    \cu{} is solvable in~$n k^{O(k)}$ time when parameterized by treewidth~$k$.
\end{observation}

We now turn towards XP-time algorithms, that is, polynomial-time algorithms for constant parameter values.
We show that \cu{} parameterized by the maximum size of an independent set is in XP.
The main idea is to first observe that the graphs induced by the vertices in~$S$ (and in~$T$, respectively) cannot have too many connected components as this would imply a large independent set in the input graph.
Next, these connected components cannot be too far apart from one another as any long induced path contains a large independent set.
Based on these two observations, it is enough to guess a small number of vertices to ensure connectivity between all vertices in~$S$ and in~$T$, respectively.
For each guess, we can then compute a minimum cut between the vertices we already colored red and blue to find an optimal solution.

\ifconf
\else
\begin{restatable}{proposition}{indset}
    \label{prop:indset}
    \cu{} parameterized by the size~$k$ of a maximum independent set in the input graph can be solved in~$O(n^{4k^2})$~time.
\end{restatable}
\fi

\appendixproof{prop:indset}{\indset*}
{
\begin{proof}
Let $(G=(V,E),S,T,\ell)$ be an instance of \cu{} where $G$ has a maximum independent set of size $k$. Let $S_1, S_2, \dots S_p$ be the connected components of $G[S]$ and let $T_1, T_2, \dots T_q$ be the connected components of $G[T]$. Note that $p \leq k$ and $q \leq k$ as we can otherwise chose one vertex from each component to get an independent set of size~$k+1$, a contradiction.
In the beginning only the vertices of $S$ are colored red and the vertices of~$T$ are colored blue.
We claim that in any solution, at most $(k-1)(2k-2)$ additional red vertices are needed to connect all vertices of $S$ and at most $(k-1)(2k-2)$ additional blue vertices are needed to connect all vertices of~$T$.
To show this, consider any solution.
To connect two connected components $S_i$ and $S_j$ of $G[S]$, there has to exist a path between $S_i$ and $S_j$ such that all vertices of the path are colored red.
Note that this also implies that there is an induced path between them where all vertices are colored red.
It is now possible to bound the length of this path by the size of the maximum independent set.
Any induced path between $S_i$ and $S_j$ contains at most $2k$ vertices as any induced path of length $2k+1$ contains an independent set of size $k+1$. Hence, to connect $S_i$ and $S_j$ in any solution at most $2k-2$ additional red vertices are needed.
Since there are $p$ connected components and $p \leq k$, at most $(k-1)(2k-2)$ vertices have to be colored red to make $G[R]$ connected. The argument for the vertices of $T$ is analogous. 

Our algorithm now guesses two sets~$R', B' \subseteq V$, each of size at most~${(k-1)(2k-2)}$.
We discard any guess where~$T \cap R' \neq \emptyset$ or $S \cap B' \neq \emptyset$. Additionally, we discard all guesses where~$G[S\cup R']$ or $G[T\cup B']$ are not connected. 
Finally, we color all remaining vertices of~$V$ by computing a minimum $S\cup R'$-$T \cup B'$-cut. 
If there are at most~$\ell$ multicolored edges in~$G$, then~$(G,\ell)$ is a yes-instance.

We next show that the running time is in~$O(n^{4^k})$.
Note that for~$k=1$, we do not need to guess any vertices and therefore the running time is~$O(m^{1+o(1)}) \subseteq O(n^4)$ as we only need to find a minimum cut \cite{B+23}.
For~$k \geq 2$, we need to guess~$2(k-1)(2k-2) = 4k^2 - 8k + 4 \leq 4k^2 - 4$ vertices.
Hence, there at most $n^{4k^2-4}$ possible guesses and for each guess finding a minimum cut can be done in $m^{1+o(1)} \in O(n^4)$ time. 
Thus, the overall runtime is in~$O(n^{4k^2})$.
\end{proof}
}

Finally, we show that \cu{} parameterized by clique-width is in XP.
We note that it remains open whether this parameter allows for fixed-parameter tractability.
Our algorithm is an adaptation of an algorithm for \maxcu{} due to Duarte et al.\;\cite{D+21}.
They use dynamic programming over a~$k$-expression of the input graph where they store the size of a largest cut with exactly~$s_i$ red vertices of each label~$i$ under certain connectivity conditions.
As in the case of treewidth, we can replace a maximum in their calculation by a minimum to solve a version of \cu{} without any terminals.
We can then incorporate terminals by first modifying the~$k$-expression into a~$3k$-expression where all vertices in~$S$ of label~$i$ get label~$i_S$ instead and all vertices in~$T$ of label~$i$ get label~$i_T$ instead.
We then simply discard any table entry in which the number of red vertices of label~$i_S$ does not equal the number of vertices of label~$i_S$ or where there are any red vertices of label~$i_T$.
Since their algorithm runs in~$O(n^{2k+4}3^{6 \cdot 2^k})$~time, this yields the following.

\begin{observation}
    \cu{} can be solved in~$O(n^{6k+4}3^{6 \cdot 8^k})$~time when parameterized by the clique-width~$k$.
\end{observation}

\section{Parameterized Hardness}
\label{sec:hard}
In this section, we show a number of hardness (both para-NP-hardness and W[1]-hardness) results for \cu{} and \dcs.
As \dcs{} is a special case of \cu, all hardness results for the former directly translate to the latter.
First, \dcs{} remains NP-hard even in bipartite graphs of bisection width one.
The idea for the reduction is to first add a copy of the graph (without any terminals) and connect one vertex of the graph to a vertex in~$S$.
Note that all vertices in the copy can be colored red and hence the size of an optimal solution remains the same.
Next, we can make the graph bipartite by subdividing each edge once.
This results in an equivalent bipartite instance of bisection width one and shows the following.

\begin{observation}
    \dcs{} is NP-hard in bipartite graphs with bisection width one.
\end{observation}

Next, we show that \cu{} is W[1]-hard when parameterized by the clique cover number even if the size of a smallest dominating set is one.
Our reduction is based on a reduction by Bentert et al.\;\cite{BDFGK24} and we prove a couple of additional properties of the reduction that will be useful when excluding polynomial kernels for the solution size.

\ifconf
\else
\begin{restatable}{proposition}{CCN}
    \label{prop:weird}
    \cu{} is W[1]-hard when parameterized by the clique cover number of the input graph even if the input graph contains a dominating set of size one, when~$|S|=1$, and when there exist constants~$c_1 \geq 0, c_2 > |T|$ such that (i)~$\ell = c_1 + |T| (c_2 - |T|)$ and (ii) any cut that keeps any set of~$j$ terminals in~$T$ connected to at least one other terminal in~$T$ while separating the terminal in~$S$ from all terminals in~$T$ has size at least~$c_1 + j (c_2-j)$.
\end{restatable}
\fi

\appendixproof{prop:weird}{\CCN*}
{
\begin{proof}
    Our proof is based on a reduction by Bentert et al.\;\cite{BDFGK24}.
    We first summarize their reduction, which shows that \cu{} is W[1]-hard when parameterized by the number of terminals even if~$|S|=1$.
    Starting from an instance~$(G,k)$ of \textsc{Regular Multicolored Clique}, a version of \textsc{Multicolored Clique} where each vertex in the input graph has degree exactly~$d$ for some~$d$, they construct an equivalent instance~$(H,S,T,\ell)$ of \cu{} as shown next and depicted in \cref{fig:weird}.
    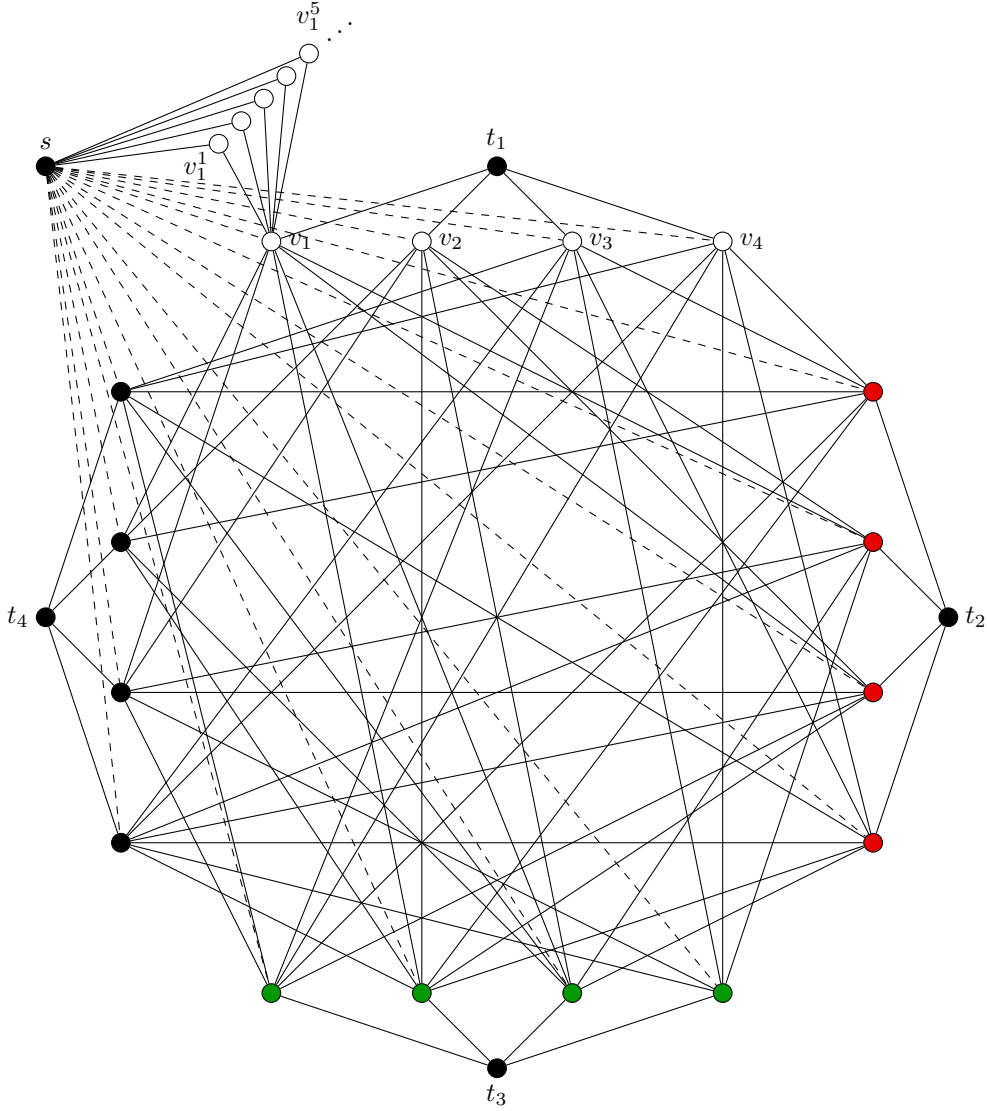
\begin{figure}[t!]
    \centering
	\begin{tikzpicture}
		\newcommand{\colorA}{blue}
		\newcommand{\colorB}{red!90!black}
		\newcommand{\colorC}{green!60!black}
		\newcommand{\colorD}{black}
		
		\foreach \x in {1,...,4}{
			\node[circle,draw, label=right:$v_\x$] (v-1\x) at (2*\x,0) {};
			\node[circle,draw,fill=\colorB] (v-2\x) at (10,-2*\x) {};
			\node[circle,draw,fill=\colorC] (v-3\x) at (2*\x,-10) {};
			\node[circle,draw,fill=\colorD] (v-4\x) at (0,-2*\x) {};
		}
		\foreach \i\j\x\y in {1/1/2/2,1/1/2/3,1/1/3/2,1/1/3/3,1/1/4/2,1/1/4/3,1/2/2/2,1/2/2/3,1/2/3/2,1/2/3/3,1/2/4/2,1/2/4/3,
							  1/3/2/1,1/3/2/4,1/3/3/1,1/3/3/4,1/3/4/1,1/3/4/4,1/4/2/1,1/4/2/4,1/4/3/1,1/4/3/4,1/4/4/1,1/4/4/4,
							  2/1/3/1,2/1/3/2,2/1/4/1,2/1/4/2,2/2/3/3,2/2/3/4,2/2/4/3,2/2/4/4,2/3/3/1,2/3/3/2,2/3/4/3,2/3/4/4,2/4/3/3,2/4/3/2,2/4/4/4,2/4/4/1,
							  3/1/4/1,3/1/4/3,3/2/4/2,3/2/4/4,3/3/4/1,3/3/4/2,3/4/4/3,3/4/4/4}{
			\draw (v-\i\j) -- (v-\x\y);
		}
			\node[circle,draw,fill=black, label=$t_1$] (t1) at(5,1) {};
			\node[circle,draw,fill=black, label=right:$t_2$] (t2) at(11,-5) {};
			\node[circle,draw,fill=black, label=below:$t_3$] (t3) at(5,-11) {};
			\node[circle,draw,fill=black, label=left:$t_4$] (t4) at(-1,-5) {};
			\foreach \x in {1,...,4}{
				\foreach \y in {1,...,4}{
					\draw (t\x) to (v-\x\y);
				}
			}
		
  	\node[circle,draw,fill=black, label=$s$] (s) at(-1,1) {};
			\foreach \x in {1,2,3,4}{
				\foreach \y in {1,2,3,4}{
					\draw[dashed] (s) to (v-\x\y);
				}
			}
        \foreach \x in {1,2,...,5}{
            \node[circle,draw] at(1+.3*\x,1+.3*\x) {} edge(s) edge(v-11);
        }
        \node at(2.9,2.9) {$\iddots$};
        \node at(2.5,3) {$v_1^5$};
        \node at(1,1) {$v_1^1$};
	\end{tikzpicture}
    \caption{Illustration of the reduction by Bentert et al.. The dashed edges represent $n+2m$ parallel~$P_3$'s as indicated between~$s$ and~$v_1$.}
    \label{fig:weird}
\end{figure}
    They start with an induced copy of~$G$, set~$S=\{s\}$ and~$T=\{t_1,t_2,\ldots,t_k\}$, and make~$t_i$ adjacent to all vertices of color~$i$ in~$G$.
    Moreover, they add vertices~$v^j_i$ for each vertex~$v_i$ in~$G$ and each~$j \in [n+2m]$, where~$n$ and~$m$ are the number of vertices and edges in~$G$, respectively.
    Each vertex~$v_i^j$ is made adjacent to~$v_i$ and~$s$.
    Finally, they set~$\ell = n - k + k(n+2m) + k(d-k+1)$.

    We next modify the construction to make~$s$ adjacent to all vertices in the induced copy of~$G$ in~$H$ and adjacent to all vertices in~$T$.
    Next, we assume that there are the same number~$\nicefrac{n}{k}$ of vertices of each color in~$G$ as this version of the problem is also known to be NP-hard.
    We further modify the reduction to make each vertex~$v_i$ in~$H$ adjacent to all other vertices that have the same color (in~$G$) and all vertices of the form~$v_i^j$ pairwise adjacent.
    We adjust the value of~$\ell = 2n + k(n+2m+d-k+1)$.
    Note that the new graph~$H$ has a dominating set of size one (the vertex~$s$ is adjacent to everything) and has a clique cover number of at most~$k+1$.
    The latter holds as all vertices of one color in~$G$ form a clique together with one vertex from~$T$ and all other vertices ($s$ and all vertices~$v_i^j$) form one clique~$C_s$.
    Moreover, observe that~$\ell = c_1 + |T|(c_2-|T|)$ for~$c_1 = 2n$ and~$c_2=(n+2m+d+1)$.
    
    It remains to prove that the adjusted instance is equivalent to the input instance of \textsc{Regular Multicolored Clique} and that any cut that keeps a set of~$j$ terminals in~$T$ connected to at least one other terminal in~$T$ while separating all terminals in~$T$ from~$s$ has size at least~$c_1 + j(c_2-j)$.
    We start with the latter.
    First note that any cut of size at most~$\ell$ cannot cut through the clique~$C_s$.
    Moreover, in order for any set of~$j$ vertices~$T' \subseteq T$ to be connected to at least one other vertex in~$T$, we need to add at least one neighbor of~$t_i$ to the connected component of~$t_i$ for each~$t_i \in T'$.
    Since the neighborhoods of two vertices~$t_i$ and~$t_j$ are by construction disjoint, this also implies that we need to separate at least~$j$ of the original vertices in~$G$ from~$C_s$.
    Moreover, any such cut is minimized when we pick exactly one neighbor of each~$t_i$ as each neighbor is adjacent to~$n+2m$ vertices in~$C_s$ and only~$\nicefrac{n}{k}+d < n+2m$ vertices outside of~$C_s$. 
    So consider any cut that separates all of~$T$ and~$j$ additional vertices from~$C_s$.
    Observe that the size of the cut is at least~$n + k - j + j(n+2m+d-j+1+\nicefrac{n}{k})$ as the vertices in~$T$ are incident to~$n+k$ edges and all but~$j$ of them are in the cut.
    Moreover, each of the~$j$ additional vertices are adjacent to~$\nicefrac{n}{k}-1$ vertices of the same color (which are not separated from~$C_s$), at least~$d-j+1$ vertices of different colors that are not separated from~$C_s$, and~$n+2m+1$ vertices in~$C_s$.
    Since we may assume that~$\nicefrac{n}{k} > k$ (otherwise there is a simple~$k^k$ time algorithm), it holds that~$n + k - j + j(n+2m+d-j+1+\nicefrac{n}{k}) \leq 2n+j(n+2m+d+1-j) = c_1 + j(c_2-j)$.
    
    To show that the constructed instance is equivalent to the original instance of \textsc{Regular Multicolored Clique}, we closely follow the proof by Bentert et al.\;\cite{BDFGK24}.
    If the constructed instance is a yes-instance, then we separate exactly one neighbor of each vertex in~$T$ from~$C_s$ as shown above.
    Let~$C$ be the set of these vertices.
    In order for a cut between~$T \cup C$ and the rest of the graph to be of size at most~$\ell$, we show that~$C$ needs to induce a clique in~$G$.
    First, each vertex in~$T$ is incident to exactly~$\nicefrac{n}{k}$ edges in the cut, that is, the vertices in~$T$ are incident to exactly~$n$ edges in the cut.
    This leaves a remaining budget of~$n + k(n+2m+d-k+1)$ for edges incident to the vertices in~$C$.
    Hence, each of these vertices can be adjacent to~$(n+2m+d-k+1+\nicefrac{n}{k})$ edges in the cut on average.
    Note that every vertex is adjacent to exactly~$n+2m+1$ vertices in~$C_s$ and to exactly~$\nicefrac{n}{k}-1$ vertices of the same color (which are all not separated from~$C_s$).
    Hence, on average each vertex in~$C$ can be incident to at most~$d-k+1$ edges to vertices of other colors in~$G$ that are not separated from~$C_s$.
    As all of these edges are exactly the edges in~$G$ and since~$G$ is~$d$-regular, each vertex needs to be incident to~$k-1$ edges to other vertices in~$C$, that is, $C$ needs to be a clique of size~$k$.

    If there is a (multicolored) clique~$C$ of size~$k$ in~$G$, then consider the cut in~$H$ between~$T \cup C$ and the rest of the graph.
    Each vertex in~$T$ is incident to~$\nicefrac{n}{k}+1$ edges and adjacent to one other vertex in~$T \cup C$.
    Hence, the cut contains~$k (\nicefrac{n}{k})$ edges incident to vertices in~$T$.
    Moreover, it contains~$k (\nicefrac{n}{k}-1+d-(k-1))$ edges between the original vertices in~$G$ in~$H$ and~$k (n+2m+1)$ edges between vertices in~$C$ and vertices in~$C_s$.
    In total, the cut has size~$k(2\nicefrac{n}{k}+d-k+n+2m+1) = 2n + k(n+2m+d-k+1) = \ell$.
    This concludes the proof.
\end{proof}
}

Finally, we show that \dcs{} remains NP-hard on graphs of maximum degree three based on a reduction due to van ’t Hof et al.\;\cite{HPW09}.
Note that both \dcs{} and \cu{} become trivial on graphs of maximum degree at most~$2$ as the input graph is then restricted to a path or a cycle.

\ifconf
\else
\begin{restatable}{proposition}{md}
    \label{prop:md}
    \dcs{} is NP-hard even if the input graph has maximum degree three.
\end{restatable}
\fi
\appendixproof{prop:md}{\md*}
{
\begin{proof}
We give a reduction from 3,4-SAT, which is known to be $\NP$-hard \cite{Tov84}.
Therein, one is given a Boolean formula~$\phi$ in which each clause contains at most three literals and each variable appears in at most four clauses.
Let $X=\{x_1,x_2,\dots, x_n\}$ and~$\overline{X}= \{\overline{x} \mid x \in X \}$ be a set of all positive and negative literals in~$\phi$ and let $C=\{C_1,C_2,\dots, C_m\}$ be the set of clauses in~$\phi$.
We construct an instance~$(G,S,T)$ of \dcs{} as follows.
For each clause~$C_i$, we add a vertex~$c_i$ to~$G$ and for each variable~$x_i$, we add paths~$P_i=(x_{i}^{1},x_{i}^{2},\dots, x_{i}^{6})$ and~$\overline{P}_i = (\overline{x}_{i}^{1},\overline{x}_{i}^{2}, \dots , \overline{x}_{i}^{6})$ to~$G$.
If a literal~$x_i \in X \cup \overline{X}$ appears in the clause $C_j$, then we add an edge between $c_j$ and one vertex of $\{x_{i}^{2},x_{i}^{3},x_{i}^{4}, x_{i}^{5}\}$ which is not adjacent to any other vertex~$c_{j'}$.
Such a vertex always exists as each variable appears in at most four clauses.
Finally, we add edges $\{x_i^6, \overline{x}_{i+1}^1\}$, $\{\overline{x}_{i}^6, x_{i+1}^1\}$, $\{x_{i}^6, x_{i+1}^1\}$ and $\{\overline{x}_{i}^6,\overline{x}_{i+1}^1\}$ for all~$i \in [n-1]$ and vertices~$f_1$ and~$f_2$ with edges~$\{f_1,x_1^1\}$, $\{f_1,\overline{x}_1^1\}$, $\{f_2,x_n^6\}$ and~$\{f_2,\overline{x}_n^6\}$. 
We conclude the construction by setting~$S=\{f_1,f_2\}$ and~$T=C$.
See \cref{fig:max_deg} for an illustration.
\begin{figure} 
    \centering
    \begin{tikzpicture}
        \foreach \x in {1,...,6}
                {\node (y\x) at (0.5*\x,5) [circle,draw,label=above:$x_{1}^\x$] {};}
        \foreach \x in {1,...,5}{
            \pgfmathtruncatemacro\y{\x+1};
                \draw (y\x) to (y\y);
                }
        \foreach \x in {1,...,6}
                {\node (y'\x) at (0.5*\x,3.5) [circle,draw,label=below:$\overline{x}_{1}^{\x}$] {};}
        \foreach \x in {1,...,5}{
            \pgfmathtruncatemacro\y{\x+1};
                \draw (y'\x) to (y'\y);
                }
        \foreach \x in {1,...,6}{
            \pgfmathtruncatemacro\z{\x+7};
                \node (z\x) at (0.5*\z,5) [circle,draw,label=above:$x_{2}^{\x}$] {};}
        \foreach \x in {1,...,5}{
            \pgfmathtruncatemacro\y{\x+1};
                \draw (z\x) to (z\y);
                }
        \foreach \x in {1,...,6}{
            \pgfmathtruncatemacro\z{\x+7};
                \node (z'\x) at (0.5*\z,3.5) [circle,draw,label=below:$\overline{x}_{2}^{\x}$] {};}
        \foreach \x in {1,...,5}{
            \pgfmathtruncatemacro\y{\x+1};
                \draw (z'\x) to (z'\y);
                }

        \foreach \x in {1,...,6}{
            \pgfmathtruncatemacro\z{\x+16};
                \node (n\x) at (0.5*\z,5) [circle,draw,label=above:$x_{n}^{\x}$] {};}
        \foreach \x in {1,...,5}{
            \pgfmathtruncatemacro\y{\x+1};
                \draw (n\x) to (n\y);
                }
        \foreach \x in {1,...,6}{
            \pgfmathtruncatemacro\z{\x+16};
                \node (n'\x) at (0.5*\z,3.5) [circle,draw,label=below:$\overline{x}_{n}^{\x}$] {};}
        \foreach \x in {1,...,5}{
            \pgfmathtruncatemacro\y{\x+1};
                \draw (n'\x) to (n'\y);
                }

        \draw (y6) to (z1);
        \draw (y'6) to (z1);
        \draw (y6) to (z'1);
        \draw (y'6) to (z'1);

        \draw[dashed] (z6) -- (n1);
        \draw[dashed] (z'6) -- (n'1);
        
        \draw[dashed] (z6) -- (7.1,4.1);
        \draw[dashed] (z'6) -- (7.1,4.4);

        \draw[dashed] (n1) -- (7.9,4.1);
        \draw[dashed] (n'1) -- (7.9,4.4);

        \node (f1) at (-0,4.25) [circle,draw,fill=fred2,draw=fred,label=left:$f_1$] {};
	    \node (f2) at (11.5,4.25) [circle,draw,fill=fred2,draw=fred,label=right:$f_2$] {};

            \draw (f1) to (y1);
            \draw (f1) to (y'1);
            \draw (n6) to (f2);
            \draw (n'6) to (f2);

        
        \node (c1) at (3,7) [circle,draw,fill=fblue2,draw=fblue,label=above:$c_1$] {};  
        \node (c2) at (5,7) [circle,draw,fill=fblue2,draw=fblue,label=above:$c_2$] {};
        \node at(7,7) {$\dots$};
        \node (c3) at (9,7) [circle,draw,fill=fblue2,draw=fblue,label=above:$c_m$] {};

            \draw[bend right=10] (c1) to (y2);
            \draw[bend left=7] (c1) to (z2);
            \draw[bend left=20] (c1) to (n'2);

            \foreach \x in {1,...,3}
               {\draw (c2) to (0.5*\x+4,6.25);}
            \foreach \x in {1,...,3}
               {\draw (c3) to (0.5*\x+8,6.25);}
  
        \end{tikzpicture}
    \caption{An illustration of the reduction behind \cref{prop:md} with~$C_1 = (x_1 \lor x_2 \lor \overline{x}_n)$.}
    \label{fig:max_deg}
\end{figure}

We next show that the maximum degree in~$G$ is three.
Note hat each vertex $c_i$ has degree at most three as each clause in~$\phi$ contains at most three variables.
For each path~$P_i$, the vertices $x_i^j$ with $2 \leq  j \leq 4$ have degree at most 3 because each $x_i^j$ is only adjacent to $x_i^{j-1}$ and $x_i^{j+1}$ and at most one vertex of $C$.
The vertices~$x_i^1$ and~$x_i^6$ have degree 3 since they are not adjacent to any vertices in $C$ and only have edges to~$x_{i-1}^6, \overline{x}_{i-1}^6,x_i^2$ and~$x_{i+1}^1,\overline{x}_{i+1}^1, x_i^5$, respectively (or to~$f_1$ or~$f_2$ in the case of~$x_1^1$, $\overline{x}_1^1$, $x_{n}^6$, and~$\overline{x}_{n}^6$).
The vertices~$f_1$ and~$f_2$ each only have degree two and the entire graph has therefore maximum degree three.

Since the reduction can clearly be computed in polynomial time, it only remains to show that the constructed instance is a yes-instance of \dcs{} if and only if~$\phi$ is satisfiable.
To this end, first assume that the constructed instance is a yes-instance and let~$(R,B)$ be a coloring of $G$ such that $G[R]$ and $G[B]$ are connected and $S \subseteq R$ and $T \subseteq B$.
We set each variable~$x_i$ to true if and only if any vertex of~$P_i$ is colored blue.
We next prove that this is a satisfying assignment. 
The vertices~$f_1$ and~$f_2$ are both colored red.
Any path that connects~$f_1$ to~$f_2$ in $G[R]$ cannot contain any vertices of~$C$ as~$C=T$.
Hence, any path between $f_1$ and $f_2$ has to contain all vertices of~$P_i$ or all vertices of~$\overline{P}_i$ for each~$i \in [n]$.
As a result, all vertices of~$P_i$ or all vertices of~$\overline{P}_i$ are colored red.
Since the vertices of $C=T$ are independent, one neighbor~$x$ of each vertex~$c_j$ has to be colored blue by the solution.
If $x = x_i^p$ for some~$i \in [n]$ and~$2 \leq p \leq 5$, then we set~$x_i$ by construction to true and~$c_j$ is satisfied.
If $x = \overline{x}_i^p$ for some~$i\in [n]$ and~$2 \leq p \leq 5$, then coloring one vertex of $\overline{P}_j$ blue forces all vertices of $P_j$ to be colored red and thus $x_i$ was by construction set to false.
The vertex~$c_i$ is only adjacent to a vertex of $\overline{P}_j$ if $\overline{x}_i$ appears in $C_i$ and therefore setting~$x_j$ to false satisfies~$C_i$.
This proves that~$\phi$ is satisfiable.

Now assume that there is a truth assignment~$\beta$ that satisfies~$\phi$.
We color all vertices of~$P_i$ blue and all vertices of~$\overline{P}_i$ red if $x_i$ is set to true by~$\beta$ and we color all vertices of~$P_i$ red and all vertices of~$\overline{P}_i$ blue, otherwise.
Next, we color the vertices of~$C$ blue and~$f_1$ and~$f_2$ red.
In this coloring exactly one path of $P_i$ and $\overline{P}_i$ is colored red and the other is colored blue. All paths~$P_i$ and~$\overline{P}_i$ are connected to~$P_{i-1}, \overline{P}_{i-1},P_{i-1},$ and~$\overline{P}_{i+1}$ (which the exception of $P_1,\overline{P}_{1}$ and $P_n,\overline{P}_{n}$ which are connected to~$f_1$ and~$f_2$ instead).
Hence, all red vertices in the paths are connected and all blue vertices in the paths are connected as well.
The vertices~$f_1$ and~$f_2$ are by construction adjacent to one red vertex in one of the paths
Finally for each $c_i \in C$, since~$\beta$ satisfies~$C_i$, there has to be a blue neighbor of~$c_i$ in one of the paths.
This shows that both~$G[R]$ and~$G[B]$ are connected in the constructed coloring and this concludes the proof.    
\end{proof}
}

\section{Polynomial Kernels}
In this section, we analyze which parameters allow for polynomial kernels.
Note that we can restrict our attention to parameters that allow for fixed-parameter tractability as this is equivalent to having a kernel of any size \cite{DF13}.
We first show that \cu{} does not admit a polynomial kernel when parameterized by the vertex cover number of the input graph plus the number of terminals.
Note that for this parameter, \dcs{} has a simple kernel of size~$O(k^3)$.
We can 2-approximate a vertex cover in polynomial time by repeatedly taking both endpoints of any uncovered edge into the solution.
We can then, for each pair of vertices in this vertex cover, mark~$k$ non-terminal vertices in the common neighborhood (or all such vertices if there are less than~$k$).
Removing all unmarked non-terminal vertices that do not belong to the approximated vertex cover results in a cubic kernel as we keep at most~$2k$ vertices in the approximate vertex cover and at most~$(2k)^2 \cdot 
 k = 4k^3$ marked vertices.

\label{sec:kernel}
\begin{theorem}
    \label{thm:nopolyvc}
    \cu{} parameterized by vertex cover number plus number of terminals does not admit a polynomial kernel unless \ppoly.
\end{theorem}

\begin{proof}
    We present an OR-cross composition from \textsc{Vertex Cover} in 3-regular graphs\footnote{A graph is 3-regular if each vertex has exactly three neighbors.}, which is known to be NP-hard \cite{FSS10}.
    We start with~$t$ instances, all with the same number~$n$ of vertices and the same solution size~$k$.
    Note that since all graphs have the same number of vertices and are 3-regular, they also have the same number~$m$ of edges.
    
    Our OR-cross composition consists of three different gadgets.
    The first gadget is called the instance-selection gadget and it consists of one vertex~$v_j$ for each input instance and $t$ additional dummy vertices~$D$.
    These~$2t$ vertices all form an independent set.
    Next, we add a vertex~$s$ and make it adjacent to all vertices~$v_j$ and a set~$W$ of~$4\binom{n}{2}$ vertices and make them adjacent to each vertex in the instance-selection gadget.
    We add~$s$ to~$S$ and~$W$ to~$T$. 

    The second gadget is called a vertex-selection gadget.
    We start with a clique of~${(n+1)\cdot k}$ vertices~$x_i^j$ for each~$i \in [n+1]$ and each~$j\in [k]$.
    All vertices~$x_{n+1}^j$ are contained in~$T$.
    Next, for each pair~$x_i^p$ and~$x_i^q$ with~$p \neq q \in [k]$ and each~$i \in [n]$, we add a vertex~$y_i^{p,q}$, add it to~$T$ and make it adjacent to~$x_i^p$ and~$x_i^q$.
    We do the same for each pair~$x_p^j$ and~$x_q^j$ for each~$p \neq q \in [n]$ and each~$j \in [k]$ and call the constructed terminal~$y_{p,q}^{j}$.
    Finally, for each~$j \in [k]$, we create a vertex~$z_j$, add it to~$S$, and make it adjacent to each vertex~$x_i^j$ for~$i \in [n]$.

    The third and final gadget is called the verification gadget and it consists of a clique of~$3\binom{n}{2}$ vertices~$a_i^j$ for~$i \in [\binom{n}{2}]$ and~$j\in [3]$.
    For each~$i \in [\binom{n}{2}]$, we add four additional vertices~$b_i,c^1_i,c_i^2,c^3_i$.
    The vertex~$b_i$ is added to~$S$ and the other three are added to~$T$.
    We make~$b_i$ adjacent to all vertices~$a_i^j$ for~$j\in [3]$.
    Moreover,~$c^1_i$ is adjacent to~$a_i^1$ and~$a_i^2$, vertex~$c_i^2$ is adjacent to~$a_i^1$ and~$a_i^3$, and~$c_i^3$ is adjacent to~$a_i^2$ and~$a_i^3$.

    It remains to connect the different gadgets.
    We arbitrarily order the vertices in each input instance and assign them numbers in~$[n]$.
    Next, we pick an arbitrary bijection~$f$ between numbers in~$[\binom{n}{2}]$ and pairs~$\{p,q\}$ with~$p \neq q \in [n]$.
    For each~$i \in [\binom{n}{2}]$, the vertex~$a_i^1$ is adjacent to~$v_j$ if and only if~$f(i)$ is \emph{not} an edge in the instance corresponding to~$v_j$.
    Moreover,~$a_i^1$ is adjacent to some arbitrary vertices in~$D$ to ensure that it has exactly~$t$ neighbors in the instance-selection gadget.
    The vertices~$a_i^2$ and~$a_i^3$ are adjacent to the~$t$ vertices in~$D$.

    Next, for each~$i\in [\binom{n}{2}]$, we make~$a_i^1$ adjacent to~$x^j_{n+1}$ in the vertex-selection gadget for all~$j \in [k]$.
    Let~$f(i) = \{p,q\}$ with~$p < q$.
    Then, we make~$a_i^2$ adjacent to~$x^j_p$ and~$a_i^3$ adjacent to~$x_q^j$ for all~$j \in [k]$.
    Finally, we set
    $\ell = 2\binom{n}{2}^2 + \binom{n}{2}(t+k+7) + k^2n + k(3n+2k-4) + t - m - 1$.

    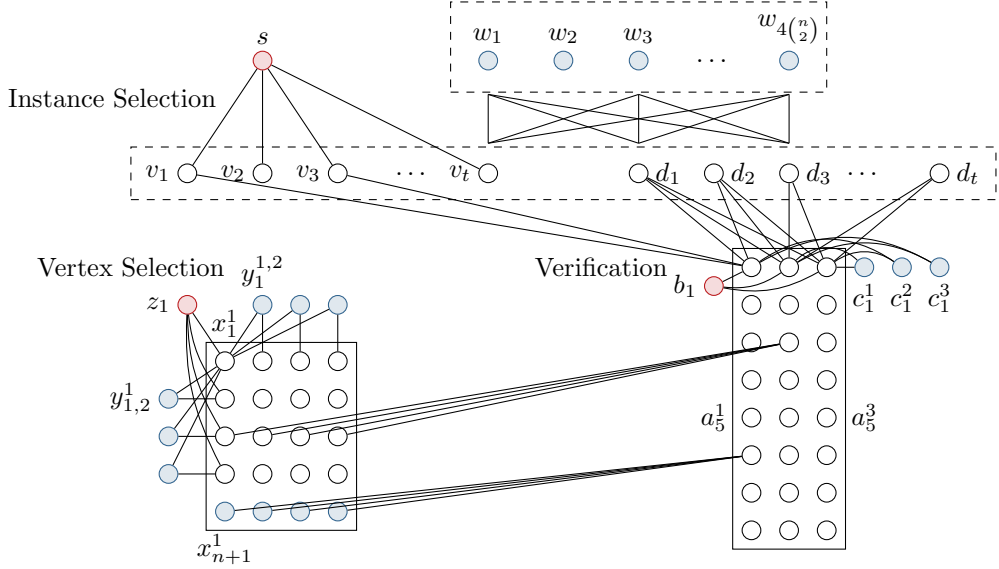
\begin{figure}
        \centering
        \begin{tikzpicture}
                \foreach \x in {1,...,3}{
					\node[circle,draw,label=left:{$v_{\x}$}] (v\x) at (\x,0) {};
                    \node[circle,draw,label=right:$d_{\x}$] (d\x) at (\x+6,0) {};
                }
                \node at (4,0) {$\dots$};
                \node at (10,0) {$\dots$};
				\node[circle,draw,label=left:{$v_t$}] (vt) at (5,0) {};
                \node[circle,draw,label=right:$d_t$] (dt) at (11,0) {};
                \node[rectangle,draw,dashed,minimum width=11.5cm,minimum height=.7cm] at(6,0) {};

                \node[circle,draw=fred,fill=fred2,label=above:$s$] (s) at (2,1.5) {} edge(vt);
                 \foreach \x in {1,...,3}{
					\draw (s) -- (v\x);
                }
                                
                 \foreach \x in {1,...,3}{
					\node[circle,draw=fblue,fill=fblue2,label=above:{$w_{\x}$}] (w\x) at (\x+4,1.5) {};
                }
                \node at (8,1.5) {$\dots$};
				\node[circle,draw=fblue,fill=fblue2,label=above:{$w_{4\binom{n}{2}}$}] (wt) at (9,1.5) {};
                \node[rectangle,draw,dashed,minimum width=5cm,minimum height=1.2cm] at(7,1.68) {};
                \foreach \x in {1,3,5}{
					\foreach \y in {1,3,5}{
                        \draw (\x+4,.4) -- (\y+4,1.05);
                    }
                }
                \node at(0,1) {Instance Selection};

                \foreach \x in {1,...,4}{
					\foreach \y in {1,...,4}{
                        \node[circle,draw] (x\x\y) at (.5*\y+1,-.5*\x-2) {};
                    }
                }
				\foreach \y in {1,...,4}{
                    \node[circle,draw,draw=fblue,fill=fblue2] (xn\y) at (.5*\y+1,-4.5) {};
                }
                \node at(1.5,-2) {$x_1^1$};
                \node at(1.5,-5) {$x_{n+1}^1$};
                \node[circle,draw=fred,fill=fred2,label=left:$z_1$] at (1,-1.75) {} edge(x11) edge[bend right=15](x21) edge[bend right=15](x31) edge[bend right=15](x41);
                \node[circle,draw=fblue,fill=fblue2,label=left:$y_{1,2}^1$] at (.75,-3) {} edge(x11) edge(x21);
                \node[circle,draw=fblue,fill=fblue2] at (.75,-3.5) {} edge(x11) edge(x31);
                \node[circle,draw=fblue,fill=fblue2] at (.75,-4) {} edge(x11) edge(x41);
                \node[circle,draw=fblue,fill=fblue2,label=$y_1^{1,2}$] at (2,-1.75) {} edge(x11) edge(x12);
                \node[circle,draw=fblue,fill=fblue2] at (2.5,-1.75) {} edge(x11) edge(x13);
                \node[circle,draw=fblue,fill=fblue2] at (3,-1.75) {} edge(x11) edge(x14);
                
                \node[rectangle,draw,minimum width=2cm,minimum height=2.5cm] at(2.25,-3.5) {};
                \node at(.25,-1.25) {Vertex Selection};
                    
                \foreach \x in {1,...,8}{
					\foreach \y in {1,...,3}{
                        \node[circle,draw] (a\x\y) at (.5*\y+8,-.5*\x-.75) {};
                    }
                }
                \node[circle,draw=fred,fill=fred2, label=left:$b_1$] at (8,-1.5) {} edge(a11) edge[bend right=20](a12) edge[bend right=20](a13);
                \node[circle,draw=fblue,fill=fblue2,label=below:$c_1^3$] at (11,-1.25) {} edge[bend right=30](a11) edge[bend right=30](a12);
                \node[circle,draw=fblue,fill=fblue2,label=below:$c_1^2$] at (10.5,-1.25) {} edge[bend right=35](a11) edge[bend right=35](a13);
                \node[circle,draw=fblue,fill=fblue2,label=below:$c_1^1$] at (10,-1.25) {} edge[bend right=45](a12) edge(a13);
                
                \node[rectangle,draw,minimum width=1.5cm,minimum height=4cm] at(9,-3) {};
                \node at(8,-3.25) {$a_5^1$};
                \node at(10,-3.25) {$a_5^3$};
                \node at(6.5,-1.25) {Verification};

				\foreach \y in {1,2,3}{
                    \draw (a12) -- (d\y);
                    \draw (a13) -- (d\y);
                }
                \draw (a12) -- (dt);
                \draw (a13) -- (dt);
                \draw (a11) -- (v1);
                \draw (a11) -- (v3);
                \draw (a11) -- (d1);
                \draw (a11) -- (d2);
				\foreach \y in {1,...,4}{
                    \draw (a61) -- (xn\y);
                    \draw (a32) -- (x3\y);
                }
            \end{tikzpicture}
        \caption{An illustration of the construction in the proof of \cref{thm:nopolyvc}. The solid boxes indicate cliques and the dashed boxes indicate independent sets.}
        \label{fig:nopolyvc}
    \end{figure}

    Note that the instance-selection gadget contains an independent set of size~$2t$ which does not contain any terminals.
    There are only~$1 + 4\binom{n}{2} + k(n+2) + k\binom{n}{2} + n \binom{k}{2} + 7\binom{n}{2} < 14n^3$ other vertices.
    Thus, the vertex cover number and the number of terminals of the resulting graph is in~$O(n^3)$.
    Since the instance can be computed in polynomial time (in~$n+t$), it only remains to show that the constructed instance is a yes-instance if and only if at least one of the input instances of \textsc{Vertex Cover} is a yes-instance.

    To this end, first assume that one of the~$t$ instances of \textsc{Vertex Cover} is a yes-instance, that is, the respective graph~$G_j$ contains a vertex cover~$K=\{p_1,p_2,\ldots,p_k\}$ of size~$k$.
    Let~$v_j$ be the vertex in the instance-selection gadget corresponding to that instance.
    We show that the constructed instance is also a yes-instance by describing the connected component~$C_S$ containing~$S$.
    This component contains all the vertices of~$S$ as well as~$v_j$,~$x_{p_i}^i$ for each~$i \in [k]$, and the following vertices of the verification gadget.
    For each~$i \in \binom{n}{2}$, if~$f(i)=\{p,q\}$ (with~$p < q$) is not an edge in~$G_j$, then we add~$a_i^1$ to~$C_S$.
    Otherwise, at least one of the endpoints~$p$ or~$q$ is contained in~$K$ and we add~$a_i^2$ if~$p \in K$ and~$a_i^3$ if~$p \notin K$.
    Note that we added exactly one vertex of each column in the vertex-selection gadget and exactly one vertex in each row of the verification gadget.
    It remains to show that the cut between~$C_S$ and the rest of the graph is of size at most~$\ell$.

    We first analyze the sizes of cuts within each of the gadgets.
    Note~$s$ is incident to exactly~$t$ edges and exactly one of the neighbors ($v_j$) is contained in~$C_S$.
    Hence,~$s$ is incident to~$t-1$ edges in the cut between~$C_S$ and the rest of the graph.
    Next, observe that~$v_j$ is incident to~$4\binom{n}{2}$ vertices in~$W$ (which are not contained in~$C_S$).
    Each vertex~$z_j$ in the vertex-selection gadget is incident to~$n$ vertices exactly one of which is in~$C_S$.
    Each of the other~$k$ vertices in~$C_S$ in the vertex-selection gadget are adjacent to~$((n+1)k-1)+(n-1)+(k-1)+1$~other vertices in the vertex-selection gadget, exactly~$k$ of which belong to~$C_S$.
    Next, each vertex~$b_i$ is adjacent to exactly two vertices not contained in~$C_S$ and each vertex~$a_i^j$ in~$C_S$ is adjacent to exactly~$2\binom{n}{2}+2$ vertices in the verification gadget that are not contained in~$C_S$.
    
    We next analyze the number of edges in the cut between~$C_S$ and the rest of the graph that go between different gadgets.
    Note that there are no edges between the instance-selection gadget and the vertex-selection gadget.
    Hence, we only need to consider edges leaving the verification gadget.
    We start with the size of such a cut assuming that no vertex in the verification gadget belongs to~$C_S$.
    Note that~$v_j$ has~$\binom{n}{2}-m$ edges to the verification gadget and each of the~$k$ vertices~$x_i^j$ in the vertex-selection gadget that belong to~$C_S$ have~$n-1$ edges to the verification gadget.
    This leads to a baseline cut of size~$\binom{n}{2}-m+k(n-1)$.
    Now notice that adding any vertex in the verification gadget to~$C_S$ increases the described cut by~$(t+k)$ if none of the neighbors of the vertex in other gadgets are contained in~$C_S$ and by~$t+k-2$ if one neighbor belongs to~$C_S$.
    By construction, it never happens for a vertex in the verification gadget that two neighbors outside the verification gadget are contained in~$C_S$.
    Moreover, we constructed the solution such that one neighbor is always contained in~$C_S$.
    Hence, the overall size of the described cut is~$\binom{n}{2}-m+k(n-1)+\binom{n}{2}(t+k-2)$.
    Combined with the size of the cuts within each gadget, the total cut size is
    \begin{align*}
        &\ t-1 - m + \binom{n}{2}(t-k+5) + k(nk+3n+2k-4) + \binom{n}{2}(2\binom{n}{2}+2)\\
         =&\ 2\binom{n}{2}^2 + \binom{n}{2}(t+k+7) + k^2n + k(3n + 2k -4) + t-m-1 = \ell.
    \end{align*}
    Note that both~$C_S$ and the rest of the graph induce a single connected component each.

    For the reverse direction, suppose that the constructed instance of \cu{} is a yes-instance.
    Let~$C_S \supseteq S$ be the set of vertices in the connected component containing~$S$ after removing the edges of a solution (a cut of size at most~$\ell$).
    First, we will argue that we can assume without loss of generality that~$C_S \setminus S$ contains exactly one vertex~$v_j$, exactly one vertex from the set~$\{a_i^1,a_i^2,a_i^3\}$ for each~$i \in \binom{n}{2}$, exactly one vertex from the set~$\{x_1^j,x_2^j,\ldots,x_n^j\}$ for each~$j \in [k]$, and at most one vertex from each set~$\{x_i^1,x_i^2,\ldots,x_i^k\}$ for each~$i \in [n]$.
    Note that this also implies that~$C_S \setminus S$ contains exactly~$\binom{n}{2}+k+1$ vertices as all other vertices belong to~$S$ or~$T$.
    Assume that~$C_S \setminus S$ contains at least two vertices from the instance-selection gadget.
    If connectivity within~$C_S$ is not of concern, then removing one of the two vertices from~$C_S$ will always decrease the size of the cut as each vertex in the instance-selection gadget is adjacent to at most~$3\binom{n}{2}$ vertices in~$C_S$ but also to~$4 \binom{n}{2}$ vertices in~$W$ (which are contained in~$T$ and therefore not in~$C_S$).
    Moreover, since all vertices in the verification gadget that have neighbors in the instance-selection gadget (all~$a$-vertices) form a clique and vertices in the instance-selection gadget are only incident to such vertices and~$s$, two vertices are never required to ensure connectivity within~$C_S$.
    On the other hand, at least one vertex from the vertex-selection gadget needs to be contained in~$C_S$ in order to connect~$s \in S$ with the rest of~$S$.
    
    Next, assume that two vertices from a set~$\{a_i^1,a_i^2,a_i^3\}$ are contained in~$C_S$.
    Then, these two vertices have a common neighbor~$c_i^j$ for some~$j \in [3]$ which is contained in~$T$ but only has these two neighbors.
    This contradicts the fact that we started with some solution to \cu.
    Again, at least one such vertex needs to be included in~$C_S$ in order to connect~$b_i \in S$ with the rest of~$S$.
    The same arguments apply to the sets~$\{x_1^j,x_2^j,\ldots,x_n^j\}$ and~$\{x_i^1,x_i^2,\ldots,x_i^k\}$ with the exception that there is no vertex in~$S$ enforcing that at least one vertex of~$\{x_i^1,x_i^2,\ldots,x_i^k\}$ belongs to~$C_S$.

    We next show that the vertices~$K$ encoded by the set of vertices in~$C_S \setminus S$ in the vertex-selection gadget form a vertex cover in the instance corresponding to the vertex~$v_j \in C_S \setminus S$ in the instance-selection gadget.
    As in the forward direction, the size of the cut between~$C_S$ and the rest of the graph has size at least~$\ell = 2\binom{n}{2}^2 + \binom{n}{2}(t+k+7) + k^2n + k(3n + 2k -4) + t-m-1$ and this bound is only achieved if each vertex in~$C_S \setminus S$ in the verification gadget (each $a$ vertex in~$C_S$) has exactly one neighbor in~$C_S \setminus S$ outside the verification gadget.
    We call such a neighbor the \emph{buddy} of the vertex.
    For each~$i \in [\binom{n}{2}]$, if~$a_i^1$ is in~$C_S$ and has a buddy, then this buddy must be~$v_j$ indicating that~$f(i)$ is not an edge in~$v_j$.
    If~$a_i^2$ or~$a_i^3$ is contained in~$C_S$ and has a buddy, then at least one of the endpoints of~$f(i)$ is contained in~$K$.
    This implies that for each pair~$\{p,q\}$ of vertices it holds that~$\{p,q\}$ is not an edge in the instance corresponding to~$v_j$ or~$p$ or~$q$ is contained in~$K$, that is, $K$ is a vertex cover in this instance.
    Note that the set~$K$ has size~$k$ as~$C_S \setminus S$ does not contain two vertices from the set~$\{x_1^j,x_2^j,\ldots,x_n^j\}$ for any~$j \in [k]$.
    This concludes the proof.
\end{proof}

We can make the instance-selection gadget into a clique in the above proof without changing any of the proof details except for the fact that the size of an optimal solution increases by exactly~$2t-1$ (as we still need to include exactly one vertex of the vertex-selection gadget in the connected component containing~$S$).
This gives the following.

\begin{corollary}
    \cu{} parameterized by distance to clique plus the number of terminals does not admit a polynomial kernel unless \ppoly.
\end{corollary}

We next show that \cu{} admits a linear kernel in the feedback edge number of the graph by simple data reduction rules for vertices of degree one and two.

\ifconf
\else
\begin{restatable}{proposition}{fes}
    \label{prop:fes}
    \cu{} parameterized by feedback edge number~$k$ admits a kernel with at most~$5k$ vertices and~$6k$ edges.
\end{restatable}
\fi
\appendixproof{prop:fes}{\fes*}
{
\begin{proof}
    We present data-reduction rules that eliminate all degree-1 vertices and bound the length of maximal induced paths, that is, paths whose internal vertices have degree two in the input graph.
    Using standard arguments, this will result in a kernel with~$O(k)$ vertices and edges \cite{BDKNN20}.

    First, note that since the input graph is connected, we can assume without loss of generality that there are no isolated vertices (vertices without any neighbors).
    Next, consider a vertex~$u$ with exactly one neighbor~$v$.
    If~$u \notin S \cup T$, then we can simply remove~$u$ as in any optimal solution, we will color~$u$ with the same color as~$v$ and therefore~$u$ will not be incident to any solution edges.
    If~$u \in S \cup T$ (we assume without loss of generality in~$S$), then there are two cases.
    If~$S = \{u\}$, then the instance is trivial as if~$T = \emptyset$, then it is a yes-instance and if~$T \neq \emptyset$, then the instance is a yes-instance if and only if~$\ell \geq 1$ as deleting the edge~$\{u,v\}$ is an optimal solution (as the input graph is connected).
    If~$|S| \geq 2$, then~$u$ needs to be in the same connected component as~$v$ in any optimal solution.
    Hence, we can remove~$u$ and add~$v$ to~$S$ instead.
    If~$v$ was already contained in~$T$, then we return a trivial no-instance.

    After applying the above procedure exhaustively, we are left with a graph without any vertices of degree at most one.
    Consider any maximal induced path, that is, a path~$P=(v_0,v_1,\ldots,v_{p})$ for some~$p$ such that~$\deg_G(v_0),\deg_G(v_{p}) > 2$ and~$\deg_G(v_i)=2$ for all~$i\in [p-1]$.
    Let~$I = \{v_i \mid i \in [p-1]\}$ be the set of internal vertices of~$P$.
    We consider the following cases based on whether~$I \cap S = \emptyset$ and~$I \cap T = \emptyset$.
    If~$I \cap S = \emptyset$ and~$I \cap T = \emptyset$, then we remove all vertices in~$I$ except for~$v_1$ and add the edge~$\{v_1,v_p\}$.
    Let~$P' = (v_0,v_1,v_p)$ be the resulting maximal induced path.
    Note that if~$v_0$ and~$v_p$ are colored with the same color, then the internal vertices of~$P$ can be colored with the same color and therefore do not increase the solution size.
    If~$v_0$ and~$v_p$ are colored differently, then we need to remove exactly one of the edges in~$P$ and this also holds true for~$P'$.
    If~$I \cap S \neq \emptyset$ and~$I \cap T = \emptyset$, then we again replace~$P$ by~$P'$ but this time we add~$v_1$ to~$S$.
    In this case, we can assume without loss of generality that all vertices in~$I$ are colored red as we need to delete at least one edge for each of~$v_0$ and~$v_p$ that are colored blue and we can assume without loss of generality that these edges are~$\{v_0,v_1\}$ and/or~$\{v_{p-1},v_p\}$.
    The case where~$I \cap S = \emptyset$ and~$I \cap T\neq \emptyset$ is analogous.
    Finally, assume that~$I \cap S \neq \emptyset$ and~$I \cap T \neq \emptyset$.
    Then let~$v_i \in S$ and~$v_j \in T$ such that no vertex between~$v_i$ and~$v_j$ is contained in~$S \cup T$.
    Note that such a pair of vertices always exists.
    We assume without loss of generality that~$i < j$.
    Since some edge between~$v_i$ and~$v_j$ has to belong to any solution, we can remove the edge~$\{v_i,v_{i+1}\}$ without loss of generality and reduce~$\ell$ by one.
    Note that both~$v_i$ and~$v_{i+1}$ now have degree one and applying the above procedure for degree-1 vertices exhaustively removes all vertices in~$I$ (or produces a trivial kernel).

    Hence, we can now assume that the graph does not contain any degree-1 vertices and all maximal induced paths have at most one internal vertex.
    Using standard arguments, the graph contains now at most~$5k$ vertices and at most~$6k$ edges, where~$k$ is the feedback edge number of the constructed graph \cite[Lemma 2]{BDKNN20}.
    Note that the feedback edge number of the graph never increases in the above procedures.    
    Hence, the produced instance has at most~$5k$ vertices and~$6k$ edges, where~$k$ is the feedback edge number of the input graph.
    This concludes the proof.
\end{proof}
}

We next show that the solution size does not allow for a polynomial kernel.
Therein, we use \cref{prop:weird} to construct an OR-cross composition in which a solution to the entire instance can only consist of a solution to one of the input instances.

\ifconf
\else
\begin{restatable}{proposition}{nopolysol}
    \label{prop:nopolysol}
    \cu{} parameterized by solution size~$\ell$ does not admit a polynomial kernel unless \ppoly.
\end{restatable}
\fi
\appendixproof{prop:nopolysol}{\nopolysol*}
{
\begin{proof}
    We present an OR-cross composition from \cu, where~$|S|=1$ and there exist constants~$c_1 \geq 0, c_2 > |T|$ such that~$\ell = c_1 + |T| (c_2 - |T|)$ and any cut that keeps any set of~$j$ terminals in~$T$ connected to at least one other terminal in~$T$ while separating the terminal in~$S$ from all terminals in~$T$ has size at least~$c_1 + j (c_2-j)$.
    This variant is NP-hard by \cref{prop:weird}.
    We consider the polynomial equivalence relation where all instances in the same equivalence class have the same solution size~$\ell'$, the same constants~$c_1,c_2$, and the same number~$k = |T|$ of terminals in~$T$.
    We assume that the number~$t$ of input instances is~$2^q$ for some integer~$q$.
    Note that this can be achieved by duplicating one of the input instances at most~$t$ times.

    Given~$t$ instances~$I_1 = (G_1,S_1,T_1,\ell'), I_2=(G_2,S_2,T_2,\ell'),\ldots,I_t=(G_t,S_t,T_t,\ell')$ of the same equivalence class, we construct an instance~$(H,S,T,\ell)$ of \cu{} as follows.
    We start with~$H$ being the disjoint union of all graphs~$G_i$ and~$S$ being the union of all~$S_i$.
    Let~$w = k \log(t) + 1$.
    We replace each edge~$\{u,v\}$ in the current graph by~$w$ paths of length two, that is, we add a set~$W_{u,v}$ of~$w$ new vertices and for each~$p \in W_{u,v}$, we add the edges~$\{u,p\}$ and~$\{p,v\}$.
    Next, we add~$k$ binary trees of depth~$\log(t)$, add the roots of these trees to~$T$, and identify the~$i$\textsuperscript{th} leaf of each of these trees with one vertex in~$T_i$ in such a way that each vertex in~$T_i$ is identified with the leaf of exactly one tree.
    See \cref{fig:nopolysol} for an illustration.
    Finally, we add edges between each pair of vertices in~$S$ and set~$\ell = w\ell'+k\log(t)$.
    This concludes the construction.
    \begin{figure}
        \centering
        \begin{tikzpicture}
            \node[ellipse,draw,label=below:$G_1$, minimum width=3cm, minimum height=3.5cm] at (0,1) {};
            \node[ellipse,draw,label=below:$G_{2}$, minimum width=3cm, minimum height=3.5cm] at (4,1) {};
            \node[ellipse,draw,label=below:$G_t$, minimum width=3cm, minimum height=3.5cm] at (10,1) {};
            \node[ellipse,draw,label=below:$T_1$, minimum width=2cm, minimum height=.9cm] at (0,1.5) {};
            \node[ellipse,draw,label=below:$T_2$, minimum width=2cm, minimum height=.9cm] at (4,1.5) {};
            \node[ellipse,draw,label=below:$T_t$, minimum width=2cm, minimum height=.9cm] at (10,1.5) {};
            \node[circle,draw=fred,fill=fred2] at(0,0) (S1) {};
            \node[circle,draw=fred,fill=fred2] at(4,0) (S2) {} edge(S1);
            \node[circle,draw=fred,fill=fred2] at(10,0) (St) {} edge[bend left=8](S1) edge(S2);
            \node[circle,draw,label=left:$u$] at(-.75,1.5) (T11) {};
            \node[circle,draw,label=right:$v$] at(-1.25,.5) (v) {};
            \node[circle,draw] at(-1.3,1.1) (w1) {} edge(T11) edge(v);
            \node[circle,draw] at(-1,1) (w2) {} edge(T11) edge(v);
            \node[circle,draw] at(-.7,.9) (w3) {} edge(T11) edge(v);
            \node[ellipse,rotate=-20,draw, minimum width=1cm,minimum height=.5cm,label=left:$W_{u,v}$] at(-1,1){};
            \node[circle,draw] at(-.25,1.5) (T12) {};
            \node[circle,draw] at(.25,1.5) (T13) {};
            \node[circle,draw] at(0.75,1.5) (T14) {};
            \node[circle,draw] at(3.25,1.5) (T21) {};
            \node[circle,draw] at(3.75,1.5) (T22) {};
            \node[circle,draw] at(4.25,1.5) (T23) {};
            \node[circle,draw] at(4.75,1.5) (T24) {};
            \node at(7,1) {$\dots$};
            \node[circle,draw] at(9.25,1.5) (Tt1) {};
            \node[circle,draw] at(9.75,1.5) (Tt2) {};
            \node[circle,draw] at(10.25,1.5) (Tt3) {};
            \node[circle,draw] at(10.75,1.5) (Tt4) {};

            \foreach \x in {1,...,4}{
                \node[circle,draw=fblue,fill=fblue2] at(3*\x-2.5,5) (r\x) {};
                \node[circle,draw] at(3*\x-3,4.5) (u\x) {} edge(r\x) edge($(3*\x-3.25,4)$) edge($(3*\x-2.75,4)$);
                \node[circle,draw] at(3*\x-2,4.5) (v\x) {} edge(r\x) edge($(3*\x-2.25,4)$) edge($(3*\x-1.75,4)$);
                \node at(3*\x-3,3.75) {$\vdots$};
                \node at(3*\x-2,3.75) {$\vdots$};
                \node[circle,draw] at(3*\x-3.5,3) (x\x) {} edge(T1\x)  edge(T2\x) edge($(3*\x-3.35,3.5)$);
                \node[circle,draw] at(3*\x-1.5,3) (x\x) {} edge(Tt\x) edge($(3*\x-1.65,3.5)$);
            }
        \end{tikzpicture}
        \caption{An illustration of the reduction behind \cref{prop:nopolysol}.}
        \label{fig:nopolysol}
    \end{figure}

    Since the solution size~$\ell$ is by construction polynomially upper-bounded by~$n + \log(t)$ and the reduction can be computed in polynomial time (in~$n+t$), it only remains to show that the constructed instance is a yes-instance if and only if at least one of the input instances is a yes-instance.
    To this end, first assume that one input instance~$I_i$ is a yes-instance.
    We construct a solution of size at most~$\ell$ in the constructed instance as follows.
    For each edge~$\{u,v\}$ in the solution, we add all edges between vertices in~$W_{u,v}$ and~$u$ to our solution cut.
    Note that this cut has size at most~$w\ell'$.
    Next, for each of the~$k$ binary trees, we separate the path from the root to the vertex in~$G_i$ from the rest of the tree.
    Note that this cuts exactly one edge in each layer and hence exactly~$\log(t)$ edges per tree.
    Thus, the constructed cut has size at most~$w \ell' + k\log(t)=\ell$.
    Note that all vertices in~$T$ (the roots of the binary trees) are connected to one another through the graph~$G_i$ and are separated from all other graphs~$G_j$ with~$j \neq i$ as they are separated from all other leaves in the binary tree and within~$G_i$ they are separated from~$S_i$.
    Since all vertices in~$S$ are connected by construction, this shows that the constructed instance is a yes-instance.
    
    For the reverse direction, assume that the constructed instance is a yes-instance and let~$F$ be a solution cut of size at most~$\ell$ in~$H$.
    First note that if for some pair~$u,v$ of vertices, an edge incident to a vertex in~$W_{u,v}$ is contained in~$F$, then we can assume without loss of generality that for each vertex in~$W_{u,v}$ exactly one incident edge is cut.
    Since~$w > k\log(t)$, this means that we can assume that~$F$ contains edges incident to at most~$\ell'$ sets~$W_{u,v}$.
    Let~$F'$ be the set of edges in the~$t$ original graphs corresponding to these~$\ell'$ sets.
    Note that if some set of~$j$ terminals in~$T$ are connected through some graph~$G_i$ in~$H - F$, then they are connected through a path in their respective binary trees and the~$j$ corresponding terminals in~$G_i$ have to be separated from~$S_i$.
    By assumption, such a cut (in the original graph~$G_i$) has size at least~$c_1 + j(c_2-j)$ and in~$H$, this corresponds to a cut of size~$w(c_1+j(c_2-j))$.
    In order to connect all terminal pairs, we have to connect all~$k$ through one graph~$G_i$ as if we use at least two graphs to connect sets of size~$j_1,j_2,\ldots,j_p$ with~$\sum_{i=1}^pj_i \geq k$, then this cut has size at least
    \[\sum_{i=1}^p w(c_1+j_i(c_2-j_i)) \geq \sum_{i=1}^p w(c_1+j_i(c_2-k)) \geq w(p c_1 + k(c_2-k)) \geq  w(c_1 + 1 + k(c_2-k)) > \ell,\] a contradiction.
    If all terminals are connected through a single graph~$G_i$, then the vertices in~$T_i$ remain connected to one another in~$H-F$ and are separated from~$S_i$.
    That is, there is a set of~$\ell'$ edges in~$G_i$ to separate all vertices in~$T_i$ from~$S_i$ while keeping all vertices in~$T_i$ connected.
    Since~$|S_i|=1$, this shows that instance~$I_i$ is a yes-instance of \cu{} and this concludes the proof.
\end{proof}
}

Finally, we show that \dcs{} does not admit a polynomial kernel parameterized by bandwidth.

\ifconf
\else
\begin{restatable}{proposition}{nopolybw}
    \label{prop:nopolybw}
    \dcs{} parameterized by bandwidth does not admit a polynomial kernel unless \ppoly.
\end{restatable}
\fi
\appendixproof{prop:nopolybw}{\nopolybw*}
{
\begin{proof}
    We present an OR-cross composition from \dcs.
    Given~$t$ instances~$I_1 = (G_1,S_1,T_1), I_2=(G_2,S_2,T_2),\ldots,I_t=(G_t,S_t,T_t)$, we construct a new instance~$(H,S,T)$ as follows.
    Initially~$H$ is the disjoint union of all~$G_i$ and~$S$ and~$T$ are the unions of all~$S_i$ and~$T_i$ respectively, where we put graph~$G_{i+1}$ to the right of graph~$G_i$.
    We then use the gadget depicted in \cref{fig:nopolybw} to connect~$G_i$ and~$G_{i+1}$ for all~$i \in [t-1]$.
    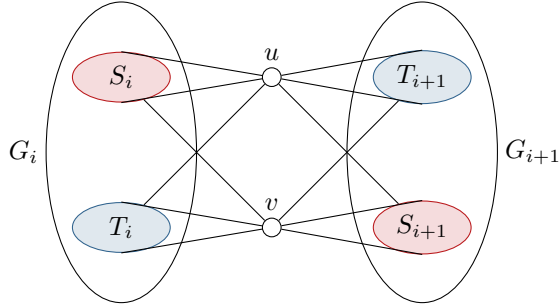
\begin{figure}[t]
        \centering
        \begin{tikzpicture}
            \node[ellipse,draw,label=left:$G_i$, minimum width=2cm, minimum height=4cm] at (0,0) {};
            \node[ellipse,draw,label=right:$G_{i+1}$, minimum width=2cm, minimum height=4cm] at (4,0) {};
            \node[ellipse,draw=fred,fill=fred2, minimum width=1.3cm, minimum height=.5cm] at(0,1) (S1) {$S_i$};
            \node[ellipse,draw=fblue,fill=fblue2, minimum width=1.3cm, minimum height=.5cm] at(0,-1) (T1) {$T_i$};
            \node[ellipse,draw=fred,fill=fred2, minimum width=1.3cm, minimum height=.5cm] at(4,-1) (S2) {$S_{i+1}$};
            \node[ellipse,draw=fblue,fill=fblue2, minimum width=1.3cm, minimum height=.5cm] at(4,1) (T2) {$T_{i+1}$};
            \node[circle,draw,label=$u$] at(2,1) {} edge(S1.north) edge(S1.south) edge(T2.north) edge(T2.south) edge(T1) edge(S2);
            \node[circle,draw,label=$v$] at(2,-1) {} edge(S2.north) edge(S2.south) edge(T1.north) edge(T1.south) edge(T2) edge(S1);
        \end{tikzpicture}
        \caption{A connection gadget in the proof of \cref{prop:nopolybw}.}
        \label{fig:nopolybw}
    \end{figure}%
    We call this gadget a \emph{connection gadget} and it simply consists of two vertices~$u$ and~$v$.
    We make~$u$ adjacent to all vertices in~$S_i$ and~$T_{i+1}$ as well as one arbitrary vertex in~$T_i$ and one arbitrary vertex in~$S_{i+1}$.
    The vertex~$v$ is adjacent to all vertices in~$T_i$ and~$S_{i+1}$ as well as one arbitrary vertex in each of~$S_i$ and~$T_{i+1}$.
    This concludes the construction.

    Since the instance can be computed in polynomial time, it only remains to show that the bandwidth of the constructed instance is polynomial in the maximum number~$n$ of vertices in any of the instances and that the constructed instance is a yes-instance if and only if at least one of the input instances is a yes-instance.
    For the former, note that placing all vertices of~$G_1$ (in any ordering) first, then the two vertices of the connection gadget between~$G_1$ and~$G_2$, then all vertices of~$G_2$ and so on results in an ordering where each edge within one of the graphs has length at most~$n-1$ and each edge incident to a vertex in the connection gadget has length at most~$n+1$.
    Since this covers all edges in the constructed instance, this shows that the bandwidth is upper-bounded by~$n+1$.

    We next show that the constructed instance is a yes-instance if and only if one of the input instances is a yes-instance.
    To this end, first assume that one instance~$I_i$ is a yes-instance.
    We construct a cut in the constructed instance as follows.
    We start with the solution cut in~$G_i$ that leaves both~$S_i$ and~$T_i$ connected and separates the two within~$G_i$.
    Next, for the connection gadget between~$G_j$ and~$G_{j+1}$ for all~$j \geq i$, we color~$u$ blue and~$v$ red.
    For all other connection gadgets, we color~$u$ red and~$v$ blue.
    Finally, in each graph~$G_j$ with~$i \neq j$, we compute a minimum~$S_j$-$T_j$-cut (ignoring connectivity) and add it to the solution.
    Note that all terminals in~$S_j$ and~$T_j$ for all~$j \neq i$ are connected to one another by a vertex in a connection gadget and also connected to one vertex of the same color in the next graph~$G_{j+1}$.
    Hence, the constructed instance is a yes-instance.

    For the reverse direction, assume that the constructed instance is a yes-instance.
    Note that in each connection gadget, any optimal solution colors both vertices of the gadget differently as otherwise at least one set of terminals is separated.
    Consider the set of connection gadgets in which~$u$ is colored red.
    If this set is empty, then we consider~$I_1$ and otherwise we consider the instance directly to the right of the last connection gadget in the set.
    We claim that the considered instance~$I_i$ is a yes-instance.
    As we assume that the~$u$ vertex in the connection gadget between~$G_{i-1}$ and~$G_i$ is colored red (if it exists) and the~$u$ vertex in the connection gadget between~$G_i$ and~$G_{i+1}$ is colored blue (again assuming it exists), each vertex in a connection gadget is adjacent to one terminal of the same color in~$G_i$, that is, they do not provide any additional connectivity between vertices in~$S_i$ and~$T_i$.
    Hence, the solution for the constructed instance contains a cut in~$G_i$ such that all vertices of~$S_i$ remain connected, are separated from all vertices in~$T_i$, which in turn remain connected to one another.
    That is, the instance~$I_i$ is a yes-instance.
    This concludes the proof.
\end{proof}
}

\section{Conclusion}
\label{sec:concl}
In this work, we studied the parameterized complexity of \cu, a natural optimization variant of \dcs.
We gave an almost complete tetrachotomy in terms of the existence of polynomial kernels, fixed-parameter tractability, and XP-time algorithms.
We conclude with a couple of open questions.
First, the complexity with respect to the distance to interval graphs remains unclear, with everything between fixed-parameter tractability and para-NP-hardness still being possibilities.
In particular, the complexity of \cu{} on interval graphs (polynomial-time solvable or NP-hard) is unknown.
Second, we showed that there is an XP-time algorithm for the parameter clique-width.
Moreover, it is known that \maxcu{} is W[1]-hard when parameterized by the clique-width. 
We conjecture that the same holds true for \cu.
Finally, it is known that \cu{} is W[1]-hard parameterized by the number of terminals (vertices in~$S \cup T$) \cite{BDFGK24}.
However, it is not known whether there is an XP-time algorithm and in particular, even whether \textsc{Network Diversion}, a special case of \cu{} with four terminals, is polynomial-time solvable or not has been a long-standing open question, which we repeat here.

\newpage
\bibliography{ref}

\begin{thebibliography}{10}

\bibitem{BDKNN20}
Matthias Bentert, Alexander Dittmann, Leon Kellerhals, Andr{\'{e}} Nichterlein,
  and Rolf Niedermeier.
\newblock An adaptive version of {B}randes' algorithm for betweenness
  centrality.
\newblock {\em Journal of Graph Algorithms and Applications}, 24(3):483--522,
  2020.

\bibitem{BDFGK24}
Matthias Bentert, Pål~Grønas Drange, Fedor~V. Fomin, Petr~A. Golovach, and
  Tuukka Korhonen.
\newblock Two-sets cut-uncut in planar graphs.
\newblock In {\em Proceedings of the 51st International Colloquium on Automata,
  Languages, and Programming ({ICALP})}, pages 22:1--22:18. Schloss Dagstuhl -
  Leibniz-Zentrum f{\"{u}}r Informatik, 2024.

\bibitem{BJK14}
Hans~L. Bodlaender, Bart M.~P. Jansen, and Stefan Kratsch.
\newblock Kernelization lower bounds by cross-composition.
\newblock {\em {SIAM} Journal on Discrete Mathematics}, 28(1):277--305, 2014.

\bibitem{B+23}
Jan {\noopsort{Brand}{van den Brand}}, Li~Chen, Richard Peng, Rasmus Kyng,
  Yang~P. Liu, Maximilian~Probst Gutenberg, Sushant Sachdeva, and Aaron
  Sidford.
\newblock A deterministic almost-linear time algorithm for minimum-cost flow.
\newblock In {\em Proceedings of the 64th {IEEE} Annual Symposium on
  Foundations of Computer Science {(FOCS)}}, pages 503--514. {IEEE}, 2023.

\bibitem{CWN13}
Christopher Cullenbine, R.~Kevin Wood, and Alexandra~M. Newman.
\newblock Theoretical and computational advances for network diversion.
\newblock {\em Networks}, 62(3):225--242, 2013.

\bibitem{Cur01}
Norman~D. Curet.
\newblock The network diversion problem.
\newblock {\em Military Operations Research}, 6(2):35--44, 2001.

\bibitem{CPPW14}
Marek Cygan, Marcin Pilipczuk, Michal Pilipczuk, and Jakub~Onufry Wojtaszczyk.
\newblock Solving the 2-disjoint connected subgraphs problem faster than $2^n$.
\newblock {\em Algorithmica}, 70(2):195--207, 2014.

\bibitem{DF13}
Rodney~G. Downey and Michael~R. Fellows.
\newblock {\em Fundamentals of Parameterized Complexity}.
\newblock Springer, 2013.

\bibitem{D+21}
Gabriel~L. Duarte, Hiroshi Eto, Tesshu Hanaka, Yasuaki Kobayashi, Yusuke
  Kobayashi, Daniel Lokshtanov, Lehilton L.~C. Pedrosa, Rafael C.~S. Schouery,
  and U{\'{e}}verton~S. Souza.
\newblock Computing the largest bond and the maximum connected cut of a graph.
\newblock {\em Algorithmica}, 83(5):1421--1458, 2021.

\bibitem{Erk02}
Ozgur Erken.
\newblock {\em A branch-and-bound algorithm for the network diversion problem}.
\newblock PhD thesis, Naval Postgraduate School, 2002.

\bibitem{FSS10}
Herbert Fleischner, Gert Sabidussi, and Vladimir~I. Sarvanov.
\newblock Maximum independent sets in 3- and 4-regular {H}amiltonian graphs.
\newblock {\em Discrete Mathematics}, 310(20):2742--2749, 2010.

\bibitem{GKLS12}
Chris Gray, Frank Kammer, Maarten L{\"{o}}ffler, and Rodrigo~I. Silveira.
\newblock Removing local extrema from imprecise terrains.
\newblock {\em Computational Geometry: Theory and Applications},
  45(7):334--349, 2012.

\bibitem{HPW09}
Pim {\noopsort{Hof}{van ’t Hof}}, Dani{\"e}l Paulusma, and Gerhard~J.
  Woeginger.
\newblock Partitioning graphs into connected parts.
\newblock {\em Theoretical Computer Science}, 410(47-49):4834--4843, 2009.

\bibitem{Kal15}
Benjamin~S. Kallemyn.
\newblock {\em Modeling Network Interdiction Tasks}.
\newblock PhD thesis, Air Force Institute of Technology, 2015.

\bibitem{KMPSL22}
Walter Kern, Barnaby Martin, Dani{\"{e}}l Paulusma, Siani Smith, and Erik~Jan
  van Leeuwen.
\newblock Disjoint paths and connected subgraphs for {$H$}-free graphs.
\newblock {\em Theoretical Computer Science}, 898:59--68, 2022.

\bibitem{LCP19}
Chungmok Lee, Donghyun Cho, and Sungsoo Park.
\newblock A combinatorial {B}enders decomposition algorithm for the directed
  multiflow network diversion problem.
\newblock {\em Military Operations Research}, 24(1):23--40, 2019.

\bibitem{NG12}
James Nastos and Yong Gao.
\newblock Bounded search tree algorithms for parametrized cograph deletion:
  Efficient branching rules by exploiting structures of special graph classes.
\newblock {\em Discrete Mathematics, Algorithms and Applications}, 4(1), 2012.

\bibitem{PR11}
Dani{\"e}l Paulusma and Johan M.~M. van Rooij.
\newblock On partitioning a graph into two connected subgraphs.
\newblock {\em Theoretical Computer Science}, 412(48):6761--6769, 2011.

\bibitem{RRS16}
Ashutosh Rai, M.~S. Ramanujan, and Saket Saurabh.
\newblock A parameterized algorithm for mixed-cut.
\newblock In {\em Proceedings of the 12th Latin American Symposium on
  Theoretical Informatics {LATIN}}, pages 672--685. Springer, 2016.

\bibitem{RS95}
Neil Robertson and Paul~D. Seymour.
\newblock Graph minors {XIII}: {T}he disjoint paths problem.
\newblock {\em Journal of Combinatorial Theory, Series B}, 63(1):65--110, 1995.

\bibitem{Sch19}
Johannes C.~B. Schröder.
\newblock Comparing graph parameters.
\newblock Bachelor's thesis, Technische Universität Berlin, 2019.

\bibitem{TV13}
Jan~Arne Telle and Yngve Villanger.
\newblock Connecting terminals and 2-disjoint connected subgraphs.
\newblock In {\em Proceedings of the 39th International Workshop on
  Graph-Theoretic Concepts in Computer Science ({WG})}, pages 418--428.
  Springer, 2013.

\bibitem{Tov84}
Craig~A. Tovey.
\newblock A simplified {NP}-complete satisfiability problem.
\newblock {\em Discrete Applied Mathematics}, 8(1):85--89, 1984.

\end{thebibliography}

\end{document}